\documentclass[journal]{IEEEtran}

\ifCLASSINFOpdf
  \usepackage[pdftex]{graphicx}
\else
  \usepackage[dvips]{graphicx}
\fi

\usepackage{cite}
\usepackage{amsmath}
\usepackage{graphicx}
\usepackage{picinpar}
\usepackage{array}
\usepackage{amssymb}
\usepackage{graphicx}
\usepackage{bm}
\usepackage{algorithm}
\usepackage{algorithmic}
\usepackage{setspace}
\usepackage{subfigure}
\usepackage{caption}
\usepackage{stfloats}
\usepackage{fixltx2e}
\usepackage{threeparttable}
\usepackage{xcolor}
\usepackage{multirow}
\usepackage{threeparttable}
\usepackage{siunitx}

\allowdisplaybreaks[4]

\newtheorem{prop}{Proposition}
\newtheorem{define}{Definition}
\newenvironment{proof}{{\noindent\it Proof}\quad}{\hfill $\square$\par}

\begin{document}

\title{Compressive Sensing Based Adaptive Active User Detection and Channel Estimation: Massive Access Meets Massive MIMO}

\author{Malong Ke,~\IEEEmembership{Student Member,~IEEE,} Zhen Gao,~\IEEEmembership{Member,~IEEE,} Yongpeng Wu,~\IEEEmembership{Senior Member,~IEEE,}\\
 Xiqi Gao,~\IEEEmembership{Fellow,~IEEE,} and Robert Schober,~\IEEEmembership{Fellow,~IEEE}
\thanks{A part of this paper was presented in the 2018 6th IEEE Global Conference on Signal and Information Processing (GlobalSIP 2018) \cite{Ke_GlobalSIP'18}.}
\thanks{M. Ke and Z. Gao are with School of Information and Electronics, Beijing Institute of Technology, Beijing 100081, China (e-mail: kemalong@bit.edu.cn; gaozhen16@bit.edu.cn).}
\thanks{Y. Wu is with Department of Electronic Engineering, Shanghai Jiao Tong University, Shanghai 200240, China (e-mail: yongpeng.wu@sjtu.edu.cn).}
\thanks{X. Gao is with the National Mobile Communications Research Laboratory, Southeast University, Nanjing 210096, China (e-mail: xqgao@seu.edu.cn).}
\thanks{R. Schober is with the Institute for Digital Communications, Friedrich-Alexander-University Erlangen-Nn$\ddot{\rm{u}}$berg, Erlangen 91054, Germany (e-mail: robert.schober@fau.de).}
}

\markboth{Journal of \LaTeX\ Class Files,~Vol.~14, No.~8, August~2018}
{Shell \MakeLowercase{\textit{et al.}}: Bare Demo of IEEEtran.cls for IEEE Journals}

\maketitle

\begin{abstract}
This paper considers massive access in massive multiple-input multiple-output (MIMO) systems and proposes an adaptive active user detection and channel estimation scheme based on compressive sensing.
By exploiting the sporadic traffic of massive connected user equipments and the virtual angular domain sparsity of massive MIMO channels, the proposed scheme can support massive access with dramatically reduced access latency.
Specifically, we design non-orthogonal pseudo-random pilots for uplink broadband massive access, and formulate the active user detection and channel estimation problems as a generalized multiple measurement vector compressive sensing problem.
Furthermore, by leveraging the structured sparsity of the uplink channel matrix, we propose an efficient generalized multiple measurement vector approximate message passing (GMMV-AMP) algorithm to realize simultaneous active user detection and channel estimation based on a spatial domain \emph {or} an angular domain channel model.
To jointly exploit the channel sparsity presented in both the spatial \emph {and} the angular domains for enhanced performance, a Turbo-GMMV-AMP algorithm is developed for detecting the active users and estimating their channels in an alternating manner.
Finally, an adaptive access scheme is proposed, which adapts the access latency to guarantee reliable massive access for practical systems with unknown channel sparsity level.
Additionally, the state evolution of the proposed GMMV-AMP algorithm is derived to predict its performance.
Simulation results demonstrate the superiority of the proposed active user detection and channel estimation schemes compared to several baseline schemes.
\end{abstract}

\vspace{-0.5mm}

\vspace{-1.5mm}
\section{Introduction}

\IEEEPARstart{V}{ideo} streaming, social networking, and the emerging Internet-of-Things (IoT) accelerate the development of base stations (BSs) that enable connectivity for billions of user equipments (UEs) with massive data volumes \cite{{Bacground:5G trends 1}}.
However, reliable massive access for massive connectivity is not supported by the current wireless networks \cite{{Bacground:5G trends 3}}.

To guarantee the availability of resources and the quality of service in massive access scenarios, uplink systems have to provide ultra-reliable low-latency detection and channel estimation (CE) for active UEs \cite{{Bacground:5G trends 1}}.
Conventional grant-based random access (RA) protocols require control signaling and the scheduling of uplink access requests for granting of resources \cite{{Grant-based protocol 1},{Grant-based protocol 2},{Grant-based protocol 3},{Grant-based protocol 4}}.
The physical random access channel (PRACH) protocol of Long-Term Evolution is one example of grant-based protocols, and can be classified into two categories: contention-free RA and contention-based RA \cite{{Grant-based protocol 1}}.
For contention-free RA, the BS first allocates UEs the dedicated preambles, which are then transmitted by the active UEs, and finally the BS responds to the requesting UEs without further contention resolution \cite{{Grant-based protocol 2}}.
For contention-based RA, multiple active UEs first transmit preambles selected from a predefined sequence set to access the BS.
Further contention resolution is required if multiple UEs choose the same preamble \cite{{Grant-based protocol 3}}.
Unfortunately, for massive access, collisions are likely to occur as the number of potential UEs can be much larger than the number of available preambles.
The authors in \cite{{Grant-based protocol 4}} proposed a strongest-user collision resolution protocol to resolve collisions in overloaded networks.
However, such grant-based solutions generally suffer from high access latency and require complicated collision resolution schemes in massive access scenarios \cite{{Grant-based protocol 5}}.

As a promising alternative, grant-free RA protocols have recently attracted significant attention, where each active UE directly transmits its pilots and data to the BS without waiting for permission \cite{{Grant-free protocol}, {Liu_SPM'18}}.
Allocating orthogonal channel resources (e.g., using orthogonal pilots as in \cite{{Orthogonal pilot}}) can facilitate the detection of the active UEs at the BS and the estimation of their channels.
However, for massive numbers of potential UEs, this approach fails due to the limited number of available orthogonal channels for a given channel coherence time.
Fortunately, a key characteristic of massive access in future wireless networks is the sporadic traffic of the UEs, i.e., out of the many potential UEs, only a small number are activated and want to access the network in any given time interval \cite{{Grant-based protocol 4}}.

Exploiting this sporadic traffic property, several compressive sensing (CS)-based grant-free RA schemes have been proposed, where the active user detection (AUD) is formulated as a sparse signal recovery problem \cite{{CS-based AUD}, {Senel_Tom'18}}.
In \cite{{CS-based MRA 1},{CS-based MRA 2}}, two advanced CS-based multi-user detection schemes, which leverage the structured sparsity over multiple time slots for accurate support detection, were proposed to jointly detect the active UEs and to decode their data.
To further exploit the a priori information about the transmitted discrete symbols for RA, the authors in \cite{{CS-based MRA 3}} developed a joint approximate message passing (AMP) and expectation maximization (EM) algorithm to improve the detection performance of sparsely active UEs.
Furthermore, the authors in \cite{{CS-based MRA 4}} proposed a threshold aided block sparsity adaptive subspace pursuit algorithm, which enables improved sparse signal recovery.
This work was also extended to cloud radio access networks (C-RANs) for efficient UE activity and data detection \cite{{CS-based MRA 5},{CS-based MRA 6}}.
However, the solutions in \cite{{CS-based MRA 1},{CS-based MRA 2},{CS-based MRA 3},{CS-based MRA 4},{CS-based MRA 5},{CS-based MRA 6}} rely on the availability of perfect channel state information (CSI), which is difficult to obtain, especially for massive wireless-connected UEs.

To jointly perform AUD and CE for single-antenna BSs, blind detection of sparse code multiple access was proposed to support grant-free RA in massive access scenarios in \cite{{CS-based MRA 7}}.
For multi-antenna systems, the authors in \cite{Park_SPAWC'17} proposed a novel joint AUD and CE scheme, where the sparsity of delay-domain channel impulse response (CIR) is leveraged for facilitating CE.
Particularly, in \cite{Park_SPAWC'17}, by iteratively exchanging the active user information and CIR estimates, an identified user cancelation technique is further employed for enhanced performance.
However, the delay-domain CIR sparsity highly depends on the channel environment and may be violated in some practical scenarios.
Considering frequency-domain CE, a modified Bayesian compressive sensing (BCS)-based access scheme for uplink C-RANs was proposed in \cite{{CS-based MRA 8}}, where the structured sparsity over multiple receive antennas are exploited.
To reduce the computational complexity, the authors in \cite{{AMP-based MRA 1}} and \cite{{AMP-based MRA 2}} developed an AMP-based scheme for AUD and CE for massive access in massive multiple-input multiple-output (MIMO) systems.
However, the solutions in \cite{{AMP-based MRA 1}} and \cite{{AMP-based MRA 2}} require the full knowledge of the a priori distribution of the channels and the noise variance, which might not be available in practice.
Besides, the work in \cite{{CS-based MRA 7},{CS-based MRA 8},{Park_SPAWC'17},{AMP-based MRA 1},{AMP-based MRA 2}} considers a narrow-band massive access scenario assuming single-carrier transmission.

Besides the aforementioned work, there are also related solutions that address the massive access problem with sparse user activity from information-theoretical perspectives \cite{{Polyanskiy_ISIT'17}, {Fengler_arXiv'19}}.
Especially in \cite{{Fengler_arXiv'19}}, the authors provided an elegant solution to support massive access with fully non-coherent detection. 

In this paper, we consider massive access for the more challenging enhanced mobile broadband (eMBB) scenario, and investigate AUD and CE for uplink massive MIMO orthogonal frequency division multiplexing (OFDM) systems.
By exploiting the sporadic traffic of the UEs and the virtual angular domain sparsity of massive MIMO channels, we develop a CS-based adaptive AUD and CE scheme.
Specifically, a pilot design based on distributed CS (DCS) theory is proposed for broadband massive access.
Moreover, the AUD and CE problems at the BS are formulated as a generalized multiple measurement vector (GMMV) CS problem \cite{{Frequency common supp 1}}.
By leveraging the structured sparsity of the uplink channel matrix, we propose a GMMV-AMP algorithm for efficient simultaneous AUD and CE based on a spatial domain or an angular domain channel model.
To further improve performance, a Turbo-GMMV-AMP algorithm is proposed for detecting the active UEs and estimating their channels in an alternating manner.
This forms the basis for an adaptive access scheme which allows the adaptation of the required access latency\footnote{For grant-free massive access, the time slots consumed for transmitting the access pilot sequence contribute to the major part of the access latency. In this paper, we focus on the time slot overhead for pilot transmission.} to the sparsity level of the uplink channel matrix, i.e., the maximum number of non-zero entries of its columns.
Additionally, the state evolution (SE) of the proposed GMMV-AMP algorithm is derived to characterize its performance.
Our main contributions can be summarized as follows.

\begin{itemize}
\item{\textbf{DCS theory-based pilot design tailored for multi-carrier systems:} Previous work \cite{{CS-based MRA 1},{CS-based MRA 2},{CS-based MRA 3},{CS-based MRA 4},
      {CS-based MRA 5},{CS-based MRA 6},{CS-based MRA 7},{CS-based MRA 8},{Park_SPAWC'17}, {AMP-based MRA 1},{AMP-based MRA 2}} mainly focuses on frequency-flat narrow-band massive access with single-carrier transmission.
      In contrast, we consider the more challenging massive access problem for eMBB, where OFDM is employed. 
      Based on DCS theory \cite{{Frequency common supp 1}}, we design pseudo-random pilot sequences tailored for multi-carrier systems.
      Thereby, the structured channel sparsity for the different subcarriers is further leveraged to improve AUD and CE performance.}

\item{\textbf{GMMV-AMP algorithm:} Exploiting the structured sparsity of the massive access channel matrix observed at multiple receive antennas and multiple subcarriers, the
      proposed GMMV-AMP algorithm facilitates efficient simultaneous AUD and CE.
      In particular, exploiting the EM algorithm, the GMMV-AMP algorithm can learn the unknown hyper-parameters of the a priori distribution of the channels and the noise variance.}

\item{\textbf{Turbo-GMMV-AMP algorithm:} To jointly leverage the channel sparsity presented in both the spatial and the angular domains, this algorithm performs AUD and
      CE in an alternating manner for further enhanced performance.
      Compared with GMMV-AMP-based simultaneous processing methods and the state-of-the-art solutions, the proposed alternating approach will reap a significant reduction of access latency for massive access.}

\item{\textbf{CS-based adaptive AUD and CE:} Most prior work \cite{{CS-based MRA 1},{CS-based MRA 2},{CS-based MRA 3},{CS-based MRA 4},{CS-based MRA 5},{CS-based MRA 6},
      {CS-based MRA 7},{CS-based MRA 8}, {Park_SPAWC'17}, {AMP-based MRA 1},{AMP-based MRA 2}} leverages the UEs' sporadic traffic only to provide a fixed access latency.
      In contrast, the access latency for the proposed adaptive access scheme can be adapted to the actual sparsity level of the uplink massive access channel matrix for reliable AUD and CE.}
\end{itemize}



\textit{Notations}: Throughout this paper, scalar variables are denoted by normal-face letters, while boldface lower and upper-case letters denote column vectors and matrices, respectively.
$\left[ {\bf{X}}_p \right]_{k,m}$ is the $(k,m)$-th element of matrix ${{\bf{X}}_p}{ \in \mathbb{C}^{K \times M}}$; $\left[ {\bf{X}}_p \right]_{k,:}$  and $\left[ {\bf{X}}_p \right]_{:,m}$ are the $k$-th row vector and the $m$-th column vector of matrix ${{\bf{X}}_p}{ \in \mathbb{C}^{K \times M}}$, respectively.
The transpose, complex conjugate, and conjugate transpose operators are denoted by ${( \cdot )^{\rm{T}}}$,  ${( \cdot )^*}$, and ${( \cdot )^{\rm{H}}}$, respectively.
$\left| {\cal K} \right|_{c}$ is the number of elements in set ${\cal K}$, $\left[ K \right]$ denotes the set $\left\{ {1, \cdots ,K} \right\}$, and $\rm{supp} \left\{  \cdot  \right\}$ is the support set of a vector or a matrix.
${\mathbb{E}}\left[  \cdot  \right]$ denotes the statistical expectation.
$\emptyset$ denotes the empty set and ${\bf{0}}_{K \times M}$ is the ${K \times M}$ zero matrix.
Finally, ${\cal C}{\cal N}\left( {x;\mu,v} \right)$ denotes the complex Gaussian distribution of a random variable $x$ with mean $\mu$ and variance $v$.

\section{System Model}

In this section, we introduce the system model of the uplink massive access in massive MIMO-OFDM systems.
Moreover, the sparsity properties of the massive access channel matrix presented in the spatial and the angular domains are further explained, respectively.

\vspace{-1mm}
\subsection{Uplink Massive Access in Massive MIMO Systems}

\begin{figure}[t]
	\centering
	\includegraphics[width=1\columnwidth,keepaspectratio]
    {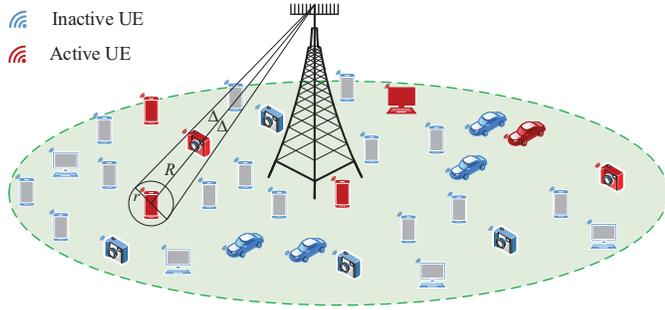}
	\caption{UEs exhibit sporadic traffic in massive access. A one-ring channel model is considered in massive MIMO systems.}
    \label{Fig:System_Model}
    \vspace{-1mm}
\end{figure}

We consider the typical uplink massive access scenario in massive MIMO systems, as illustrated in Fig. \ref{Fig:System_Model}.
There are one BS equipped with an $M$-antenna uniform linear array (ULA) and $K$ potential UEs, where $K$ is usually large (e.g., $K = 10^3$ in \cite{{AMP-based MRA 2}}).
OFDM with $N$ subcarriers is adopted to combat time dispersive channels, and $P$ pilots are uniformly allocated across the $N$ subcarriers.
For the subchannel of the $p$-th pilot subcarrier $\left(1 \le p \le P\right)$, the signal ${\bf y}_{p,k}^t \in \mathbb{C}^{M \times 1}$ received at the BS from the $k$-th UE in the $t$-th time slot (or equivalently the $t$-th OFDM symbol) can be expressed as
\begin{equation}
{\bf y}_{p,k}^t = {{\bf h}_{p,k}}{s_{p,k}^t} + {\bf n}_p^t,
\end{equation}
where ${\bf h}_{p,k} \in \mathbb{C}^{M \times 1}$ is the subchannel associated with the $k$-th UE, $s_{p,k}^t$ is the uplink access pilot of the $k$-th UE, and ${\bf n}_p^t$ denotes the additive white Gaussian noise (AWGN) at the BS for the $p$-th pilot subcarrier and the $t$-th time slot.
Here, without loss of generality, we consider single-antenna UEs.
For a typical massive access scenario, within a given time interval, only a small number of UEs are activated to access the BS.
The UE activity indicator is denoted as $\alpha_k$, and is equal to 1 when the $k$-th UE is active and 0 otherwise.
Meanwhile, we define the set of active UEs as ${\cal A} = \left\{k | {\alpha_k} = 1,\;1 \le k \le K \right\}$, and the number of active UEs is denoted by $K_a = \left|\cal A\right|_c$.
Hence, the signal received at the BS from all active UEs for the $p$-th pilot subcarrier and the $t$-th time slot is given as follows
\begin{equation}
{\bf y}_p^t = \sum_{k = 1}^K {\alpha_k}{{\bf h}_{p,k}}{s_{p,k}^t} + {\bf n}_p^t = {{\bf H}_p}{{\bf s}_p^t} + {\bf n}_p^t,
\end{equation}
where ${\bf H}_p = \left[{\alpha_1}{{\bf h}_{p,1}}, \cdots ,{\alpha_K}{{\bf h}_{p,K}}\right] \in \mathbb{C}^{M \times K}$ and ${\bf s}_p^t = \left[s_{p,1}^t, \cdots, s_{p,K}^t\right]^{\rm T} \in \mathbb{C}^{K \times 1}$.
By considering both the large-scale and the small-scale fading, we can model ${\bf h}_{p,k}$ as ${\bf h}_{p,k} = {\rho_k}{\widetilde {\bf h}_{p,k}}$, where $\rho_k$ is the large-scale fading caused by path loss and shadowing, and $\widetilde {\bf h}_{p,k}$ is the small-scale fading.
For the $p$-th pilot subcarrier, the subchannel of the $k$-th UE is modeled as follows \cite{{Ch_Vir}}
\vspace{-1mm}
\begin{equation}\label{Eq:Channel_Model}
{\widetilde {\bf h}_{p,k}} = \sum_{l = 1}^L {\beta_{k,l}}{{\bf a}_R}\left(\phi_{k,l}\right)e^{-j2\pi{\varpi_{k,l}}\left( - \frac{B_s}{2} + \frac{{B_s}\left(pN/P - 1\right)}{N}\right)},
\end{equation}
where $N/P$ is an integer, $L$ denotes the number of multi-path components (MPCs), $\beta_{k,l}$ and $\varpi_{k,l}$ are the complex path gain and the path delay of the $l$-th MPC, respectively, and $B_s$ is the two-sided bandwidth.
The array response vector ${{\bf a}_R}\left(\phi_{k,l}\right)$ is given by ${{\bf a}_R}\left(\phi_{k,l}\right) = \left[1, e^{-j2\pi{\phi_{k,l}}}, \cdots, e^{-j2\pi\left(M - 1\right){\phi_{k,l}}}\right]^{\rm T}$, where $\phi_{k,l} = \frac{d}{\lambda}{\sin\left(\varphi_{k,l}\right)}$.
Here, $\varphi_{k,l}$ is the angle of arrival (AOA) of the $k$-th UE's $l$-th MPC, $\lambda$ is the wavelength, and $d=\lambda/2$ is the antenna spacing.

\subsection{Space-Frequency Structured Sparsity in Massive Access}
\label{Sub:Spa_Fre}

\begin{figure}[!tp]
    \centering
    \subfigure[]{\includegraphics[width=1.72in]{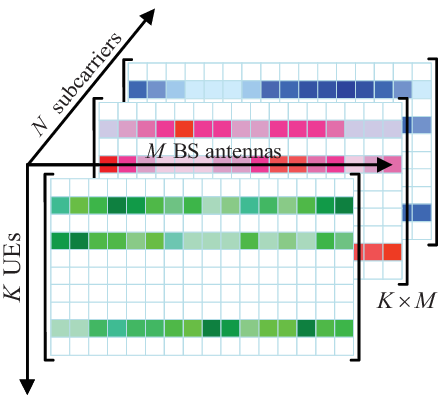}}
    \subfigure[]{\includegraphics[width=1.72in]{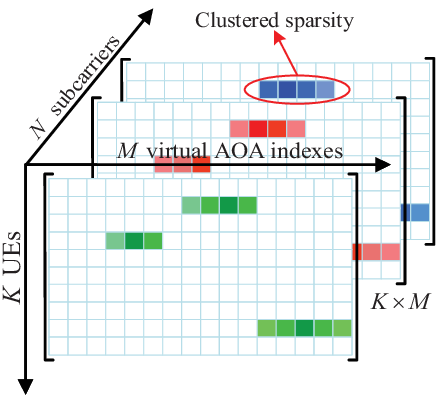}}
    \caption{The uplink massive access channel matrix exhibits two forms of structured sparsity: (a) Space-frequency structured sparsity due to sparse UE activity; (b) Angular-frequency structured sparsity due to the limited angular spread of the MPCs.}
    \label{Fig:Struct_Sparse}
    \vspace{-2mm}
\end{figure}

Due to the sporadic traffic of the UEs, only a small number of UEs are active, i.e., $K_a \ll K$.
Thus, by defining ${\bf X}_p = {\bf H}_p^{\rm T}$, the channel vector $\left[{\bf X}_p\right]_{:,m} \in \mathbb{C}^{K \times 1}$ observed at the $m$-th receive antenna for the $p$-th pilot subcarrier is sparse as
\begin{equation}\label{Eq:Spa_Ch_Sparse}
\left|{\rm supp}\left\{\left[{\bf X}_p\right]_{:,m}\right\}\right|_c = K_a \ll K.
\end{equation}
Moreover, all BS antennas exhibit the same sparsity,
\begin{equation}\label{Eq:Spa_Sparse}
{\rm supp}\left\{\left[{\bf X}_p\right]_{:,1}\right\} = {\rm supp}\left\{\left[{\bf X}_p\right]_{:,2}\right\} =  \cdots = {\rm supp}\left\{\left[{\bf X}_p\right]_{:,M}\right\}.
\end{equation}
We refer to this property as the spatial domain structured sparsity of massive access.
Since the $\alpha_k, \forall k$, are identical for all subchannels, the $\left\{{\bf X}_p\right\}_{p = 1}^P$ also exhibit a common sparsity pattern in the frequency domain as follows
\begin{equation}\label{Eq:Spa_Fre_Sparse}
{\rm supp}\left\{{\bf X}_1\right\} = {\rm supp}\left\{{\bf X}_2\right\} = \cdots = {\rm supp}\left\{{\bf X}_P\right\}.
\end{equation}
The joint structured sparsity in (\ref{Eq:Spa_Sparse}) and (\ref{Eq:Spa_Fre_Sparse}) is referred to as the space-frequency structured sparsity of $\left\{ {{{\mathbf{X}}_p}} \right\}_{p = 1}^P$.
To illustrate this structured sparsity in Fig.\ref{Fig:Struct_Sparse}(a), as an example, we assume that $K_a = 3$ active UEs out of $K = 10$ total UEs access the BS, which is equipped with $M = 16$ antennas.

\subsection{Angular-Frequency Structured Sparsity in Massive MIMO}
\label{Sub:Vir_Fre}

On the other hand, the BS is usually at high elevation with few scatterers around, whereas the UEs are typically located at low elevation in a rich local scattering environment far from the BS \cite{{Frequency common supp 1}}.
This scenario can be modeled using the classical one-ring channel model \cite{{One-ring}}.
For a UE which is located at a distance of $R$ from the BS and surrounded by rich scatterers located within a radius of $r$ around the UE, the angular spread $\Delta  \approx {\rm arctan}\left(r/R\right)$ seen from the BS is expected to be very small as usually $R \gg r$.
This leads to sparsity of massive MIMO channels in the virtual angular domain \cite{{A_R}, {Lin_CL'17}, {Yuan_Tcom'18}}.
Specifically, the virtual angular domain massive MIMO channel associated with the $k$-th UE for the $p$-th pilot subcarrier can be represented as
\begin{equation}\label{Eq:Spa_to_Vir}
{\widetilde {\bf w}_{p,k}} = {{\bf A}_R^{\rm H}}{\widetilde {\bf h}_{p,k}},
\end{equation}
where the transformation matrix ${\bf A}_R \in \mathbb{C}^{M \times M}$ at the BS side is a unitary matrix.
Here, ${\bf A}_R$ depends on the geometry of the adopted array, and becomes the discrete Fourier transform matrix for a ULA with $d = \lambda /2$ \cite{{Ch_Vir}}.
Due to the small $\Delta$ and large $M$, the channel vector $\widetilde {\bf w}_{p,k}$ is sparse, i.e.,
\begin{equation}\label{Eq:Vir_Sparse}
\left|{\rm supp}\left\{\widetilde {\bf w}_{p,k}\right\}\right|_c \ll M,
\end{equation}
and this sparsity is clustered, as illustrated in Fig. \ref{Fig:Struct_Sparse}(b).
Moreover, since the spatial propagation characteristics of all wireless channels within the total bandwidth are similar, all subchannels associated with different subcarriers are affected by the same scatterers \cite{{Frequency common supp 1}}.
Consequently, the $\left\{\widetilde {\bf w}_{p,k}\right\}_{p = 1}^P, \forall k$, have a common sparsity pattern in the frequency domain, i.e.,
\begin{equation}\label{Eq:Vir_Fre_Struct}
{\rm supp}\left\{\widetilde {\bf w}_{1,k}\right\} = {\rm supp}\left\{\widetilde {\bf w}_{2,k}\right\} =  \cdots  = {\rm supp}\left\{\widetilde {\bf w}_{P,k}\right\}.
\end{equation}
We refer to the jointly structured sparsity in (\ref{Eq:Vir_Sparse}) and (\ref{Eq:Vir_Fre_Struct}) as the angular-frequency structured sparsity of massive MIMO channels.
Additionally, we further define the virtual angular domain channel matrix as ${\bf W}_p = {{\bf X}_p}{{\bf A}_R^*} = \left[{\alpha_1}{{\bf w}_{p,1}}, \cdots, {\alpha_K}{{\bf w}_{p,K}}\right]^{\rm T}$, where ${\bf w}_{p,k} = {\rho_k}{\widetilde {\bf w}_{p,k}}$.
Considering the space-frequency structured sparsity described in (\ref{Eq:Spa_Ch_Sparse})-(\ref{Eq:Spa_Fre_Sparse}), we further have $\left|{\rm supp}\left\{\left[{\bf W}_p\right]_{:,m}\right\}\right|_c \ll K_a$, and
\begin{equation}\label{Eq:Vir_Fre_Sparse}
{\rm supp}\left\{{\bf W}_1\right\} = {\rm supp}\left\{{\bf W}_2\right\} = \cdots = {\rm supp}\left\{{\bf W}_P\right\}.
\end{equation}
The illustration in Fig. \ref{Fig:Struct_Sparse}(b) takes both the space-frequency structured sparsity of massive access and the angular-frequency structured sparsity of massive MIMO channels into account.
These sparsity properties will be exploited in the remainder of this paper to achieve low-latency and highly-reliable AUD and CE performance.

\section{CS-Based Active User Detection and Channel Estimation Schemes}
\label{Sec:Proposed Schemes}

In this section, we detail the proposed CS-based AUD and CE schemes for massive access.
First, a DCS-based pilot design is proposed for broadband massive access.
Then, two simultaneous AUD and CE schemes are developed based on a spatial domain and an angular domain channel model, respectively.
On this basis, the alternating AUD and CE schemes are further proposed for enhanced performance.
Finally, the computational complexity of the proposed schemes will be analyzed.

The frame structure of the uplink signals is illustrated in Fig. \ref{Fig:Frame_Structure}.
A frame consists of $T$ time slots, where the first $G$ time slots include both pilots and data, and the remaining $(T - G)$ time slots are reserved for data transmission only.
Here, we assume $T$ is smaller than the channel coherence time, and the activity of the UEs during the $T$ time slots remains unchanged.
At the BS, the received signals in $G$ successive time slots for the $p$-th pilot subcarrier are collected as
\begin{equation}\label{Eq:CS_Model_Spa}
{\bf Y}_p^G = {{\bf S}_p^G}{{\bf X}_p} + {\bf N}_p,\;  \forall p \in \left[P\right],
\end{equation}
where ${\bf Y}_p^G = \left[{\bf y}_p^1, \cdots, {\bf y}_p^G\right]^{\rm T} \in {\mathbb C}^{G \times M}$, ${\bf S}_p^G = \left[{\bf s}_p^1, \cdots, {\bf s}_p^G\right]^{\rm T} \in {\mathbb C}^{G \times K}$, ${\bf X}_p = {\bf H}_p^{\rm T} \in {\mathbb C}^{K \times M}$, and ${\bf N}_p = \left[{\bf n}_p^1, \cdots, {\bf n}_p^G\right]^{\rm T}$.
To avoid complicated scheduling protocols and the associated latencies, in RA, the active UE set (AUS) $\cal A$ and the corresponding channel vectors $\left\{{\bf h}_{p,k}\right\}_{p = 1}^P$, $k \in {\cal A}$, have to be reliably estimated based on the noisy measurements $\left\{{\bf Y}_p^G \right\}_{p = 1}^P$ and the known pilot matrices $\left\{{\bf S}_p^G\right\}_{p = 1}^P$, which is equivalent to estimating $\left\{{\bf X}_p\right\}_{p = 1}^P$ based on (\ref{Eq:CS_Model_Spa}).
Due to the space-frequency structured sparsity of $\left\{{\bf X}_p\right\}_{p = 1}^P$, the AUD and CE based on (\ref{Eq:CS_Model_Spa}) can be formulated as a CS problem with $G \ll K$.
Moreover, considering the virtual angular domain sparsity shown in (\ref{Eq:Vir_Sparse}), we can further transform (\ref{Eq:CS_Model_Spa}) as
\begin{equation}\label{eq:CS Model Vir}
{\bf R}_p^G = {{\bf Y}_p^G}{{\bf A}_R^*} = {{\bf S}_p^G}{{\bf W}_p} + {\widetilde {\bf N}_p},\; \forall p \in \left[P\right],
\end{equation}
where ${\widetilde {\bf N}_p} = {{\bf N}_p}{{\bf A}_R^*}$.
Based on this, we develop two categories of AUD and CE schemes:
\begin{itemize}
\item{\textbf{Simultaneous AUD and CE:}}
The estimate of $\left\{{\bf X}_p\right\}_{p=1}^P$, denoted by  $\{\widehat {\bf X}_p\}_{p=1}^P$, can be directly obtained based on (\ref{Eq:CS_Model_Spa}).
Alternatively, we can first estimate $\left\{{\bf W}_p\right\}_{p=1}^P$ based on (\ref{eq:CS Model Vir}), and then obtain $\{{\widehat{\bf X}_p}\}_{p=1}^P$ according to (\ref{Eq:Spa_to_Vir}).

\item{\textbf{Alternating AUD and CE:}}
Compared with ${\bf X}_p$ in (\ref{Eq:CS_Model_Spa}), the sparser ${\bf W}_p$ in (\ref{eq:CS Model Vir}) will yield a better CE performance, but the common sparsity pattern across multiple columns of ${\bf W}_p$ is destroyed.
Therefore, the proposed alternating scheme leverages (\ref{Eq:CS_Model_Spa}) for AUD and (\ref{eq:CS Model Vir}) for CE, i.e., (\ref{Eq:CS_Model_Spa}) and (\ref{eq:CS Model Vir}) are alternately exploited to reap both the structured sparsity of ${\bf X}_p$ and the enhanced sparsity of ${\bf W}_p$ for further improved performance.
\end{itemize}
In the following, we will first discuss the pilot design and then explain the proposed AUD and CE schemes.

\begin{figure}[t]
	\centering
	\includegraphics[width=1\columnwidth,keepaspectratio]
    {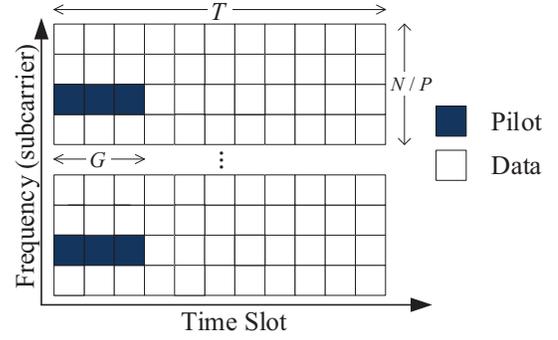}
	\caption{The frame structure of the uplink signals}
    \label{Fig:Frame_Structure}
    \vspace{-1mm}
\end{figure}

\subsection{DCS-Based Pilot Design for Broadband Massive Access}
\label{Sub:Pilot Design}

For the CS problems in (\ref{Eq:CS_Model_Spa}) and (\ref{eq:CS Model Vir}), the RA pilot matrices ${\bf S}_p^G$, $\forall p$, serve as measurement matrix.
The properties of the measurement matrix are crucial for guaranteeing reliable recovery of sparse channel matrices $\left\{{\bf X}_p\right\}_{p=1}^P$.
Hence, the pilot signals should be carefully designed to guarantee reliable AUD and CE.
The sparse signal recovery algorithms proposed in this paper are based on the family of AMP algorithms, which usually require independent and identically distributed (i.i.d.) Gaussian measurement matrices \cite{{AMP}}.
Hence, for the $p$-th pilot subcarrier, the pilot associated with the $k$-th UE in the $t$-th time slot is generated from a standard complex Gaussian distribution, i.e., $s_{p,k}^t \sim {\cal CN}(s_{p,k}^t; 0, 1)$.
Furthermore, the ${\bf S}_p^G$ should be different for different pilot subcarriers to achieve diversity, which means that (\ref{Eq:CS_Model_Spa}) and (\ref{eq:CS Model Vir}) are GMMV-CS models \cite{{Frequency common supp 1}}.
Compared to the conventional multiple measurement vector (MMV) problem \cite{MMV}, where identical pilots would be allocated to all pilot subcarriers, employing different pilot matrices across different pilot subcarriers can improve AUD and CE performance according to DCS theory \cite{{Frequency common supp 1}}.

\vspace{-3mm}
\subsection{Simultaneous AUD and CE Schemes}
\label{Sub:S_AUD_CE}

For massive access, reliable inference of AUS $\cal A$ and the corresponding $\left\{{\bf h}_{p,k}\right\}_{p = 1}^P$, $k \in {\cal A}$, from $\left\{{\bf Y}_p^G \right\}_{p = 1}^P$ is challenging.
In this subsection, we propose two simultaneous AUD and CE schemes, where a GMMV-AMP algorithm is developed to solve the CS problems in (\ref{Eq:CS_Model_Spa}) and (\ref{eq:CS Model Vir}), respectively.
Without loss of generality, we consider (\ref{Eq:CS_Model_Spa}) first and focus on the $p$-th pilot subcarrier.
The obtained results can be easily extended to the model (\ref{eq:CS Model Vir}) and multiple pilot subcarriers.

\emph{1) Spatial Domain Simultaneous AUD and CE (Scheme 1):}
Define $x_{p,k,m} = \left[{\bf X}_p\right]_{k,m}$.
The minimum mean square error (MMSE) estimate of ${\bf X}_p$ is the posterior mean, which can be expressed as
\begin{equation}\label{Eq:Post_Mean_Int}
{\hat x_{k,m}} = \int {x_{k,m}}p\left(x_{k,m}|{\bf Y} \right)d{x_{k,m}},\; \forall k,m.
\end{equation}
In (\ref{Eq:Post_Mean_Int}), the superscript $G$ and index $p$ in $x_{p,k,m}$, ${\bf X}_p$, and ${\bf Y}_p^G$ are dropped to simplify the notation, and the marginal posterior distribution is given by
\begin{equation}\label{Eq:Mar_Post}
p\left(x_{k,m}|{\bf Y}\right) = \int p\left({\bf X}|{\bf Y}\right)d{\bf X}_{\backslash k,m},
\end{equation}
where ${\bf X}_{\backslash k,m}$ denotes the collection of the $\left\{x_{i,j}\right\}_{1 \le j \le M, j \ne m}^{1 \le i \le K, i \ne k}$.
The joint posterior distribution in (\ref{Eq:Mar_Post}) can be computed according to the Bayesian rule as
\begin{equation}\label{Eq:Joint_Post}
\begin{aligned}
p\left({\bf X}|{\bf Y}\right) & = \frac{p\left({\bf Y}|{\bf X}\right)p_0\left({\bf X}\right)}{p\left({\bf Y}\right)} \\
                              & = \frac{1}{\widetilde Z_1}\prod\limits_{m = 1}^M{\left[\prod\limits_{g = 1}^G{p\left(y_{g,m}|{\bf X}\right)}\prod\limits_{k = 1}^K{p_0\left(x_{k,m}\right)}\right]},
\end{aligned}
\end{equation}
where ${\widetilde Z_1} = \iint p\left({\bf Y}|{\bf X}\right)p_0\left({\bf X}\right)d{\bf X}d{\bf Y}$ is a normalization factor and $p_0\left({\bf X}\right)$ is the a priori distribution of ${\bf X}$.
Under the assumption of AWGN, the likelihood function in (\ref{Eq:Joint_Post}) is
\begin{equation}\label{Eq:Likelihood}
p\left(y_{g,m}|{\bf X}\right) = \frac{1}{\pi\sigma}\exp\left( - \frac{1}\sigma\left|y_{g,m} - \sum\nolimits_k {s_{g,k}}{x_{k,m}}\right|^2\right), \!
\end{equation}
where $\sigma$ is the variance of the complex AWGN.
In this paper, to characterize the sparsity of ${\bf X}$, we consider a flexible spike and slab a priori distribution for ${\bf X}$, i.e.,
\setcounter{equation}  {16}
\begin{equation}\label{Eq:Spike_Slab}
\begin{aligned}
\!\!\!p_0\left({\bf X}\right) & = \prod\limits_{m = 1}^M{\prod\limits_{k = 1}^K{p_0\left(x_{k,m}\right)}} \\
                              & = \prod\limits_{m = 1}^M{\prod\limits_{k = 1}^K{\left[\left(1 - \gamma_{k,m}\right)\delta\left(x_{k,m}\right) + {\gamma_{k,m}}f\left(x_{k,m}\right)\right]}},\!\!\!
\end{aligned}
\end{equation}
which can effectively capture the actual prior knowledge of the channel matrix $\bf X$ \cite{{A_R}, {Lin_CL'17}, {Yuan_Tcom'18}}.
In (17), $0 < \gamma_{k,m} <1$ is the sparsity ratio, i.e., the probability of $x_{k,m}$ being non-zero, $\delta\left(\cdot\right)$ is the Dirac delta function, and $f\left(\cdot\right)$ is the distribution of the non-zero entries.
This distribution arises from the literature of AMP algorithm \cite{{Vila_TSP'13}, {Meng_CM'18}}, and has been successfully employed in various AMP-based channel estimation schemes \cite{{A_R}, {Lin_CL'17}, {Yuan_Tcom'18}}.

The factorization in (\ref{Eq:Joint_Post}) can be represented by a bipartite graph, which consists of variable nodes, factor nodes, and the corresponding edges \cite{{Sum-Product}}.
This suggests the use of message passing algorithms \cite{{Sum-Product}} to realize the MMSE estimator.
As the messages for marginal posterior probabilities are difficult to compute for massive access, we resort to the AMP algorithm \cite{{AMP}}, which employs low-complexity heuristics for approximating $p\left(x_{k,m}|{\bf Y}\right)$.


\begin{prop}
In the large system limit, i.e., $K \rightarrow \infty$, while $\gamma = K_a/K$ and $\kappa = G/K$ are fixed, the AMP algorithm decouples the matrix estimation problem based on (\ref{Eq:CS_Model_Spa}) into $KM$ scalar estimation problems.
Considering this, the posterior distributions of $x_{k,m}$, $\forall k,m$, are approximated as
\begin{equation}\label{Eq:Post_Approx1}
\begin{aligned}
p\left(x_{k,m}|{\bf Y}\right) &\approx p\left(x_{k,m}|C_{k,m}^q, D_{k,m}^q\right) \\
                              &\approx \frac{1}{\widetilde Z_2}p_0\left(x_{k,m}\right){\cal CN}\left(x_{k,m}; C_{k,m}^q, D_{k,m}^q\right),
\end{aligned}
\end{equation}
where $q$ denotes the $q$-th iteration, and ${\widetilde Z_2}$ is a normalization factor.
In (\ref{Eq:Post_Approx1}), $D_{k,m}^q$ and $C_{k,m}^q$ are updated at the variable nodes of the bipartite graph as
\begin{align}
D_{k,m}^q &= \left[\sum\nolimits_g \frac{\left|s_{g,k}\right|^2}{\sigma + V_{g,m}^q}\right]^{-1}, \label{Eq:Var_Update1} \\
C_{k,m}^q &= {\hat x_{k,m}^q} + D_{k,m}^q\sum\nolimits_g{\frac{s_{g,k}^*\left(y_{g,m} - Z_{g,m}^q\right)}{\sigma + V_{g,m}^q}}, \label{Eq:Var_Update2}
\end{align}
where $V_{g,m}^q $ and $Z_{g,m}^q$ are updated at the factor nodes of the bipartite graph as
\begin{align}
V_{g,m}^q &= \sum\nolimits_k{\left|s_{g,k}\right|^2{v_{k,m}^q}}, \label{Eq:Fac_Update1}\\
Z_{g,m}^q &= \sum\nolimits_k{s_{g,k}{\hat x_{k,m}^q} - \frac{V_{g,m}^q}{\sigma + V_{g,m}^{q - 1}}\left(y_{g,m} - Z_{g,m}^{q - 1}\right)}. \label{Eq:Fac_Update2}
\end{align}
\end{prop}

\begin{proof}
Please refer to the Appendix A.
\end{proof}

It is worth noticing that, although the proposed algorithms are developed from the large system limit ($K \rightarrow \infty$), in practice, they perform well even in the medium size problems, such as massive access with hundreds even thousands of UEs, which has been discussed in the literature of AMP algorithms [33, Sec. II], \cite{AMP_SE}.
Moreover, we consider the widely used Gaussian a priori distribution for the channel gains, i.e., $f\left(x_{k,m}\right) = {\cal CN}\left(x_{k,m}; \mu, \tau\right)$ \cite{A_R}.
By exploiting this a priori model in (\ref{Eq:Post_Approx1}), the posterior distribution of $x_{k,m}$ is obtained as
\begin{equation}\label{Eq:Post_Approx2}
\begin{aligned}
\!\!\! p\left(x_{k,m}|C_{k,m}^q, D_{k,m}^q\right) &= \left(1 - \pi_{k,m}^q\right)\delta\left(x_{k,m}\right) \\
                                                  &+ \pi_{k,m}^q{\cal CN}\left(x_{k,m}; A_{k,m}^q, B_{k,m}^q\right), \!\!\!
\end{aligned}
\end{equation}
where
\begin{align}
A_{k,m}^q &= \frac{\tau{C_{k,m}^q} + \mu{D_{k,m}^q}}{D_{k,m}^q + \tau},\; B_{k,m}^q = \frac{\tau{D_{k,m}^q}}{\tau + D_{k,m}^q}, \label{Eq:Var_Mean_Var} \\
\pi_{k,m}^q &= \frac{\gamma_{k,m}}{\gamma_{k,m} + \left(1 - \gamma_{k,m}\right)\exp\left( - {\cal L}\right)}, \label{Eq:Belief_Indicator} \\
{\cal L} &= \frac{1}{2}\ln{\frac{D_{k,m}^q}{D_{k,m}^q + \tau}} + \frac{\left|C_{k,m}^q\right|^2}{2D_{k,m}^q} - \frac{\left|C_{k,m}^q - \mu\right|^2}{2\left(D_{k,m}^q + \tau\right)}, \label{Eq:L}
\end{align}
and $\pi _{k,m}^q$ is referred to as the belief indicator.
The posterior mean (\ref{Eq:Define_Post_Mean}) and variance (\ref{Eq:Define_Post_Var}) can now be explicitly calculated as
\begin{align}
g_a\left(C_{k,m}^q, D_{k,m}^q\right) &= {\pi_{k,m}^q}{A_{k,m}^t}, \label{Eq:Post_Mean} \\
g_c\left(C_{k,m}^q, D_{k,m}^q\right) &= {\pi_{k,m}^q}\left(\left|A_{k,m}^q\right|^2 + B_{k,m}^q\right) - \left|g_a\right|^2, \label{Eq:Post_Var}
\end{align}
respectively.

Therefore, for the $p$-th pilot subcarrier, the MMSE estimate of ${\bf X}_p$ can be acquired by iteratively calculating (\ref{Eq:Var_Update1})-(\ref{Eq:Post_Var}) instead of solving the high-dimensional integrals in (\ref{Eq:Mar_Post}).
The resulting procedure is referred to as the \emph{basic MMV-AMP algorithm}.
However, the basic MMV-AMP algorithm requires full knowledge of the a priori distribution of the channels $\left\{\mu, \tau, \gamma_{k,m}, \forall k,m\right\}$ and the noise variance $\sigma$, which may be difficult to obtain in practice.
Hence, the EM algorithm is exploited to learn the unknown hyper-parameters, i.e., $\bm{\theta} = \left\{\mu, \tau, \sigma, \gamma _{k,m}, \forall k,m\right\}$.
The EM algorithm involves two steps
\begin{align}
Q\left({\bm \theta}, {\bm \theta}^q\right) &= {\mathbb E}\left[\ln{p\left({\bf X}, {\bf Y}\right)}|{\bf Y}; {\bm \theta}^q\right], \label{Eq:EM1}\\
{\bm \theta}^{q + 1} &= \arg\mathop{\max}\limits_{\bm \theta} Q\left({\bm \theta}, {\bm \theta}^q\right), \label{Eq:EM2}
\end{align}
where ${\mathbb E}\left[ \cdot |{\bf Y}; {\bm \theta}^q\right]$ denotes the expectation conditioned on measurements ${\bf Y}$ with parameters ${\bm \theta}^q$, i.e., the expectation is with respect to the posterior distribution $p\left({\bf X}|{\bf Y}; {\bm \theta}^q\right)$.
There are two challenges in harnessing the EM algorithm: (a) the computation of $p\left({\bf X}|{\bf Y}; {\bm \theta}^q\right)$ is of high complexity and (b) the joint optimization of all elements of ${\bm{\theta }}$ is difficult.
Fortunately, in the large system limit with $K \to \infty$, the high complexity of calculating $p\left({\bf X}|{\bf Y}; {\bm \theta}^q\right)$ can be considerably reduced by using the  approximation $p\left({\bf X}|{\bf Y}; {\bm \theta} ^q\right) \approx \prod\limits_m{\prod\limits_k{p\left(x_{k,m}|C_{k,m}^q, D_{k,m}^q\right)}} $ according to (\ref{Eq:Post_Approx1}).
Moreover, the incremental EM algorithm \cite{{Incremental EM}} can be used to simplify the joint optimization of all elements of ${\bm \theta}$, where ${\bm \theta}$ is updated one element at a time and the other parameters are held constant.
By setting the derivative of (\ref{Eq:EM1}) with respect to one element of ${\bm \theta}$ to zero, the update rules of the hyper-parameters are obtained as, $\forall k,m$:
\begin{align}
\mu_{k,m}^{q+1} &= \frac{\sum\nolimits_k{{\pi_{k,m}^q}{A_{k,m}^q}}}{\sum\nolimits_k{\pi_{k,m}^q}}, \label{Eq:Pri_Mean}\\
\gamma_{k,m}^{q+1} &= \pi_{k,m}^{q+1} = \frac{\gamma_{k,m}^q}{\gamma_{k,m}^q + (1-\gamma_{k,m}^q)\exp\left(-{\cal L}\right)}, \label{Eq:Sparse_Ratio}\\
\tau_{k,m}^{q+1} &= \frac{\sum\nolimits_k{{\pi_{k,m}^q}\left[\left|\mu_{k,m}^q - A_{k,m}^q\right|^2 + B_{k,m}^q\right]}}{\sum\nolimits_k{\pi_{k,m}^q}}, \label{Eq:Pri_Var}\\
\sigma_{k,m}^{q+1} &= \frac{1}{GM} \! \sum\nolimits_g \!\! {\left[\frac{\left|y_{g,m} \!-\! Z_{g,m}^q\right|^2}{\left|1 \!+\! V_{g,m}^q/\sigma_{k,m}^q\right|^2} + \frac{{\sigma_{k,m}^q}{V_{g,m}^q}}{\sigma_{k,m}^q \!+\! V_{g,m}^q}\right]}. \label{Eq:Noi_Var}
\end{align}
As the EM algorithm may converge to a local extremum of the likelihood function, the proper initialization of the hyper-parameters is crucial.
Here, we use the following initialization \cite{{EM Initial}}, $\forall k,m$:
\begin{align}
\sigma_{k,m}^1 &= \frac{\sum\nolimits_g\left|y_{g,m}\right|^2}{({\rm SNR}^0 + 1)G}, \label{Eq:Noi_Initial}\\
\mu_{k,m}^1 &= 0,\; \tau_{k,m}^1 = \frac{\sum\nolimits_g\left|y_{g,m}\right|^2 - M{\sigma_{k,m}^1}}{\sum\nolimits_g{\sum\nolimits_k{\left|s_{g,k}\right|^2}}}, \label{Eq:Pri_Mean_Var_Initial} \\
\!\!\! \gamma_{k,m}^1 &= \frac{G}{K} \! \left\{ \! \mathop{\max}\limits_{c>0} \frac{1 \!-\! 2K \! \left[(1 \!+\! c^2)\Phi(-c) \!-\! c\phi(c)\right]/G}{1 \!+\! c^2 \!-\! 2\left[(1 \!+\! c^2)\Phi(-c) \!-\! c\phi(c)\right]} \! \right\}. \!\! \label{Eq:Sparse_Ratio_Initial}
\end{align}
Here, $\Phi(-c)$ and $\phi(c)$ are the cumulative distribution function and the probability density function of the standard normal distribution, respectively, and ${\rm SNR}^0\! =\! 100$ is suggested in \cite{{EM Initial}}.

Equations (\ref{Eq:Var_Update1})-(\ref{Eq:Post_Var}) and (\ref{Eq:Pri_Mean})-(\ref{Eq:Sparse_Ratio_Initial}) are the main EM steps incorporated into the MMV-AMP algorithm to learn the unknown hyper-parameters.
However, the resulting overall algorithm is limited to the AUD and CE of a single pilot subcarrier.
Hence, we extend the MMV-AMP algorithm to the GMMV-AMP algorithm (summarized in Algorithm \ref{Alg:GMMV-AMP}), where the subchannel matrices ${\bf X}_p$, $\forall p$, for all pilot subcarriers are jointly estimated with different measurement matrices ${\bf S}_p^G, \forall p$.
Specifically, in \emph{lines \ref{Step:Fac_Update}-\ref{Step:Var_Update}}, the messages are updated independently for all pilot subcarriers; in \emph{lines \ref{Step:Damp1}} and \emph{\ref{Step:Damp2}}, a damping parameter $\rho = 0.3$ is used to prevent the algorithm from diverging according to \cite{{Damp_Rangan14}}; \emph{line} \ref{Step:Hyper_Update} uses the incremental EM algorithm to learn the unknown hyper-parameters $\bm{\theta}$; \emph{line \ref{Step:Hyper_Refine}} refines the update rule for the sparsity ratio ${\gamma}_{p,k,m}$ to leverage the structured sparsity of the channel matrix for improved CS recovery.
By contrast, the state-of-the-art AMP-based estimators in \cite{{AMP-based MRA 1}} and \cite{{AMP-based MRA 2}} require $\bm{\theta}$ as a priori information.

In Algorithm \ref{Alg:GMMV-AMP}, the sparsity ratio $\gamma_{p,k,m}$ is the probability that the $\left(k,m\right)$-th element of ${\bf X}_p$ is non-zero.
In \emph{line \ref{Step:Hyper_Update}}, $\gamma_{p,k,m}$ is updated independently for all $p$, $k$, and $m$ according to (\ref{Eq:Sparse_Ratio}), which indicates that the common sparsity pattern described in (\ref{Eq:Spa_Sparse}) and (\ref{Eq:Spa_Fre_Sparse}) is not exploited.
To fully exploit the structured sparsity of the channel matrix, as discussed in Section \ref{Sub:Spa_Fre} and illustrated in Fig. \ref{Fig:Struct_Sparse}(a), we assume that the channel elements associated with the same UE have a common sparsity ratio, and further propose to refine $\gamma_{p,k,m}$ as in \emph{line \ref{Step:Hyper_Refine}} of Algorithm \ref{Alg:GMMV-AMP}, where we use
\vspace{-1mm}
\begin{equation} \label{Eq:Refine_1}
\!\!{\cal N}_{p,k,m}\! =\! \left\{\left(o, l, u\right)|o = 1, \cdots, P; \; l = k; \; u = 1, \cdots, M\right\}.\!\!
\vspace{-4mm}
\end{equation}

With the estimate of $\left\{{\bf X}_p\right\}_{p=1}^P$, the AUS and the corresponding channel vectors can be simultaneously acquired.
Specifically, for AUD, we develop two UE activity detectors based on the $\{\widehat {\bf X}_p\}_{p=1}^P$ and the belief indicators $\pi_{p,k,m}$, $\forall p,k,m$, respectively, as follows.
We first define a threshold function $r\left(x; \varepsilon\right)$, where $r\left(x; \varepsilon\right) = 1$ if $ \left|x\right| > {\varepsilon}$, otherwise $r\left(x; \varepsilon\right) = 0$.

\begin{algorithm}[t]
\caption{GMMV-AMP Algorithm}
\label{Alg:GMMV-AMP}
\begin{algorithmic}[1]
\REQUIRE $\forall p$\;: Noisy observations ${\bf Y}_p^G$, pilot matrices ${\bf S}_p^G$; the damping parameter $\rho$, the maximum number of iterations $T_{\rm amp}$, and termination threshold $\eta$.
\ENSURE Estimated channel matrix $\{{\widehat {\bf X}_p}\}_{p=1}^P$ and the related belief indicators $\pi_{p,k,m}$, $\forall p,k,m$.
\STATE $\forall p,k,m,g$: Set iteration index $q$ to 1, initialize the hyper-parameters as in (\ref{Eq:Noi_Initial})-(\ref{Eq:Sparse_Ratio_Initial}), and initialize other parameters as $V_{p,g,m}^0 = 1$, $Z_{p,g,m}^0 = y_{p,g,m}$, $\hat x_{p,k,m}^1 = \mu_{p,k,m}^1$, $v_{p,k,m}^1 = \tau_{p,k,m}^1$.
\label{Step:Initial}
\REPEAT
\STATE $\forall p,g,m$: Update $V_{p,g,m}^q$ and $Z_{p,g,m}^q$ according to (\ref{Eq:Fac_Update1}) and (\ref{Eq:Fac_Update2}) at the factor nodes.
\label{Step:Fac_Update}
\STATE $V_{p,g,m}^q = {\rho}{V_{p,g,m}^{q-1}} + \left(1-\rho\right){V_{p,g,m}^q}$.
\label{Step:Damp1}
\STATE $Z_{p,g,m}^q = {\rho}{Z_{p,g,m}^{q-1}} + \left(1-\rho\right){Z_{p,g,m}^q}$.
\label{Step:Damp2}
\STATE $\forall p,k,m$: Update $D_{p,k,m}^q$ and $C_{p,k,m}^q$ according to (\ref{Eq:Var_Update1}) and (\ref{Eq:Var_Update2}) at the variable nodes, and $\hat x_{p,k,m}^{q+1} = {g_a}\left(C_{p,k,m}^q, D _{p,k,m}^q\right)$, $v_{p,k,m}^{q+1} = {g_c}\left(C_{p,k,m}^q, D_{p,k,m}^q\right)$.
\label{Step:Var_Update}
\STATE $\forall p,k,m$: Update the hyper-parameters $\mu_{p,k,m}^{q+1}$, $\gamma_{p,k,m}^{q+1}$, $\tau_{p,k,m}^{q+1}$, and $\sigma_{p,k,m}^{q+1}$ as in
(\ref{Eq:Pri_Mean})-(\ref{Eq:Noi_Var}).
\label{Step:Hyper_Update}
\STATE $\forall p,k,m$: Refine the update rule for the sparsity ratio, $\gamma_{p,k,m}^{q+1} = \frac{1}{\left|{\cal N}_{p,k,m}\right|_c}\sum_{(o,l,u) \in {\cal N}_{p,k,m}}{\pi _{o,l,u}^{q+1}}$.
\label{Step:Hyper_Refine}
\STATE $q = q + 1$.
\UNTIL $q\! >\!T_{\rm amp}$ or $\sum_p{\left\|{\widehat {\bf X}}_p^q\! -\! {\widehat {\bf X}}_p^{q-1}\right\|_{\rm F}} / \sum_p{\left\|{\widehat {\bf X}}_p^{q-1}\right\|_{\rm F}}\! <\! \eta$.
\RETURN $\{{\widehat {\bf X}_p}\}_{p=1}^P$;\; ${\pi_{p,k,m}} = \gamma_{p,k,m}^q$, $\forall p,k,m$.
\end{algorithmic}
\end{algorithm}

\begin{define}
Based on $\{\widehat {\bf X}_p\}_{p=1}^P$, a channel gain-based activity detector (CG-AD) is proposed for AUD as follows
\begin{equation}\label{Eq:CG-AD}
{\widehat \alpha_k} = \left\{ \begin{array}{*{20}{c}}
1, &\frac{1}{PM}{\sum_p{\sum_m{r\left({\hat x_{p,k,m}}; \varepsilon_{\rm cg}\right) \ge p_{\rm cg}}}},\\
0, &\frac{1}{PM}{\sum_p{\sum_m{r\left({\hat x_{p,k,m}}; \varepsilon_{\rm cg}\right) < p_{\rm cg}}}},
\end{array} \right.
\end{equation}
where $\varepsilon_{\rm cg} = {\widetilde \alpha}\max\left\{ \left|\hat x_{p,k,m}\right|,\; \forall p, k, m\right\}$, $\widetilde \alpha = 0.01$ is the ratio of the minimum and maximum amplitudes of the channel coefficients, see [23, Sec. IV], and $p_{\rm cg} =0.9$\footnote{If more than $90\%$ of the elements of $\{[\widehat {\bf X}_p]_{k,:}\}_{p=1}^P$ are decided to be non-zero, the $k$-th UE is declared active.}.
\end{define}

\begin{prop}
In the large system limit, if a reliable estimate of ${\bf X}_p$ is acquired after the convergence of the GMMV-AMP algorithm,
\begin{equation}
\pi_{k,m} = \left\{ \begin{array}{*{20}{c}} 1,\; x_{k,m} \ne 0, \\ 0,\; x_{k,m} = 0. \end{array}  \right.
\end{equation}
\end{prop}

\begin{proof}
Please refer to the Appendix \ref{Proof:Belief Indicator}.
\end{proof}

\begin{define}
\label{Defn:BI-AD}
Since the belief indicator $\pi_{p,k,m}$ tends to be 1 for ${\hat x}_{p,k,m} \ne 0$ and 0 for ${\hat x}_{p,k,m} = 0$, we further design a belief indicator-based activity detector (BI-AD) as
\begin{equation}\label{Eq:BI-AD}
{\widehat \alpha_k} = \left\{\begin{array}{*{20}{c}}
1, &\frac{1}{PM}{\sum_p{\sum_m{r\left(\pi_{p,k,m}; \varepsilon_{\rm bi}\right) \ge p_{\rm bi}}}},\\
0, &\frac{1}{PM}{\sum_p{\sum_m{r\left(\pi_{p,k,m}; \varepsilon_{\rm bi}\right) < p_{\rm bi}}}}.
\end{array} \right.
\end{equation}
For spatial domain simultaneous AUD and CE, we set $\varepsilon_{\rm bi} = \varepsilon_{\rm bi}^{\rm spa} = 0.5$ to make the missed detection and false alarm probabilities identical, and the same as $p_{\rm cg}$, we set $p_{\rm bi} = p_{\rm bi}^{\rm spa} = 0.9$ based on empirical experience.
Nevertheless, we note that the decisions of the CG-AD and BI-AD are not sensitive to the values of $p_{\rm cg}$ and $p_{\rm bi}$\footnote{We found empirically via simulations for a wide range of system parameters that $p_{\rm cg} = p_{\rm bi}^{\rm spa} = 0.9$ yields a high AUD performance.}.
Finally, if the $k$-th UE is declared active, its channel is estimated as ${\widehat {\bf h}_{p,k}} = \left[ {{{\widehat {\mathbf{X}}}_p}} \right]_{k,:}^{\text{T}}$.
\end{define}

\emph{2) Angular Domain Simultaneous AUD and CE (Scheme 2):}
The GMMV-AMP algorithm designed for CS model (\ref{Eq:CS_Model_Spa}) can be directly applied to model (\ref{eq:CS Model Vir}) for angular domain simultaneous AUD and CE by replacing ${\bf Y}_p^G$ and $\widehat {\bf X}_p^G$ with ${\bf R}_p^G$ and $\widehat {\bf W}_p^G$, respectively.
Actually, (\ref{Eq:CS_Model_Spa}) and (\ref{eq:CS Model Vir}) are equivalent signal models.
Meanwhile, it is worth noticing that the channel matrix ${\bf X}_p$ and ${\bf W}_p$ exhibit two different forms of structured sparsity.
Hence, the difference between the spatial domain and the angular domain schemes mainly lies in the different update rules (i.e., (\ref{Eq:Refine_1}) and (\ref{Eq:Neighbor Set})) for refining the sparsity ratio $\gamma_{p,k,m}$ in \emph{line \ref{Step:Hyper_Refine}} of Algorithm \ref{Alg:GMMV-AMP}.
Different from (\ref{Eq:CS_Model_Spa}), both the space-frequency and the angular-frequency structured sparsity of the channel matrix are considered in (\ref{eq:CS Model Vir}), as has been discussed in Section \ref{Sub:Vir_Fre} and illustrated in Fig. \ref{Fig:Struct_Sparse}(b).
Define the neighbors of $w_{p,k,m}$ as
\begin{equation}\label{Eq:Neighbor Set}
\begin{aligned}
{\widetilde {\cal N}}_{p,k,m} = &\left\{\left(p - 1, k, m\right), \left(p + 1, k, m\right), \right.\\
                     &\left. \left(p, k, m - 1\right), \left(p, k, m + 1\right) \right\}.
\end{aligned}
\end{equation}
Due to the clustered sparsity and the structured sparsity described in (\ref{Eq:Vir_Sparse})-(\ref{Eq:Vir_Fre_Sparse}), $w_{p,k,m}$ and the elements of ${\widetilde {\cal N}}_{p,k,m}$ tend to be simultaneously either zero or non-zero\footnote{Here, we assume that no power leakage caused by the discrete Fourier transformation in (\ref{Eq:Spa_to_Vir}) \cite{A_R}. Otherwise, the channel would be approximately sparse.}.
Hence, when the GMMV-AMP algorithm is applied to (\ref{eq:CS Model Vir}), ${\cal N}_{p,k,m}$ in \emph{line \ref{Step:Hyper_Refine}} is replaced by ${\widetilde{\cal N}}_{p,k,m}$.
Based on the estimate of $\left\{{\bf W}_p\right\}_{p=1}^P$, the $\{\widehat {\bf X}_p\}_{p=1}^P$ can be acquired  according to (\ref{Eq:Spa_to_Vir}).
Moveover, the active UEs can be detected via CG-AD in (\ref{Eq:CG-AD}) based on $\{\widehat {\bf X}_p\}_{p=1}^P$ or via BI-AD in (\ref{Eq:BI-AD}) based on the belief indicators of $w_{p,k,m}$, $\forall p,k,m$.
However, as the common sparsity across multiple BS antennas is destroyed, it is challenging to find a suitable $p_{\rm bi}$ for BI-AD.
Hence, we set $p_{\rm bi}$ to $p_{\rm bi}^{\rm ang} = p_{\rm bi}^{\rm spa}{S_{\rm min}^a}/M$, where $S_{\rm min}^a$ is the minimum number of non-zero elements of $\left[{\bf W}_p\right]_{k,:}$, $\forall p, k$.
For AUD, our simulations in Section \ref{Sec:Sim Result} reveal that, for \emph{Scheme 1} based on model (\ref{Eq:CS_Model_Spa}), BI-AD is more reliable than CG-AD, but for \emph{Scheme 2} based on model (\ref{eq:CS Model Vir}), CG-AD is more reliable than BI-AD .

\vspace{-3mm}
\subsection{Alternating AUD and CE Schemes}
\label{Sub:A_AUD_CE}

The simultaneous AUD and CE based on (\ref{Eq:CS_Model_Spa}) or (\ref{eq:CS Model Vir}) can not fully exploit the enhanced sparsity of ${\bf W}_p$ and the common sparsity pattern across the multiple columns of ${\bf X}_p$.
Hence, we propose a Turbo-GMMV-AMP algorithm that performs AUD and CE in an alternating manner, where model (\ref{Eq:CS_Model_Spa}) and model (\ref{eq:CS Model Vir}) are alternately exploited for enhanced performance.
This facilitates the development of a CS-based adaptive AUD and CE scheme for practical massive access scenarios with unknown channel sparsity level.

\begin{figure}[t]
	\centering
	\includegraphics[width=1\columnwidth,keepaspectratio]
    {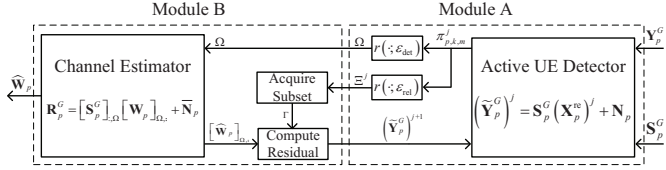}
	\caption{Block diagram of the proposed Turbo-GMMV-AMP algorithm.}
    \label{Fig:Turbo_Iteration}
    \vspace{-5mm}
\end{figure}

\emph{1) Turbo-GMMV-AMP Algorithm (Scheme 3):}
This algorithm is summarized in Algorithm \ref{Alg:Turbo-GMMV-AMP}, and consists of module A and module B, as illustrated in Fig. \ref{Fig:Turbo_Iteration}.
In the first turbo iteration ($j = 1$), module A determines a rough estimate of the AUS, where the GMMV-AMP algorithm is applied based on model (\ref{Eq:CS_Model_Spa}) to obtain the belief indicators $\pi_{p,k,m}^j$, $\forall p,k,m$.
Even for small $G$ with poor estimation quality of ${\bf X}_p$, i.e., $v_{k,m}^{\infty}$ is smaller than a specific value rather than $v_{k,m}^{\infty}\! \to\! 0$, the $0 \le \pi_{p,k,m}^j \le 1$ tend to be 0 for $x_{p,k,m} = 0$ but 1 for $x_{p,k,m} \ne 0$ after convergence of the GMMV-AMP algorithm, the related proof is similar to Appendix \ref{Proof:Belief Indicator}.
Hence, we use the BI-AD rather than CG-AD to acquire two AUS estimates with different reliability: a rough AUS $\Omega$ with a lower threshold ${\varepsilon_{\rm{det}}}$ and a reliable AUS $\Xi^j$ with a higher threshold ${\varepsilon_{\rm{rel}}}$, as shown in \emph{lines \ref{Step:Det_Start}-\ref{Step:Det_End}} in Algorithm \ref{Alg:Turbo-GMMV-AMP}, so that $\Xi^j \subseteq \Omega$.
Here, we set $\varepsilon_{\rm det}$ to 0.4 to obtain a low missed detection probability, and $\varepsilon_{\rm rel}$ to 0.9 to avoid false alarm and to achieve reliable detection of the active UEs, i.e., the UEs in $\Xi^j$ are active with high probability.
These two AUSs, $\Omega$ and $\Xi^j$, are passed on to module B.

In module B, with the rough AUS estimate $\Omega$, the angular domain channel vectors of the UEs in $\Omega$, i.e., $\left[{\bf W}_p\right]_{\Omega,:}$, are estimated based on the model in (\ref{eq:CS Model Vir}) as
\begin{equation}\label{Eq:CS_Model_Vir_Act}
{\bf R}_p^G = \left[{\bf S}_p^G\right]_{:,\Omega}\left[{\bf W}_p\right]_{\Omega,:} + {\overline {\bf N} _p},\; \forall p \in \left[P\right],
\end{equation}
where $\left[{\bf S}_p^G\right]_{:,\Omega} \in \mathbb{C}^{G \times \left|\Omega\right|_c}$ and $\left[{\bf W}_p\right]_{\Omega,:} \in \mathbb{C}^{\left|\Omega\right|_c \times M}$ are sub-matrices of ${\bf S}_p^G$ and ${\bf W}_p$, respectively, ${\overline {\bf N}_p} = \left[{\bf S}_p^G\right]_{:,{{\cal K} - \Omega}}\left[{\bf W}_p\right]_{{{\cal K} - \Omega},:} + {\widetilde {\bf N}_p}$, ${\cal K}$ is the set of all potential UEs, and ${\cal K} - \Omega$ denotes the difference set of sets $\cal K$ and $\Omega$.
To reduce the power of ${\overline {\bf N}_p}$, ${\cal A} \subseteq \Omega$ is desirable, i.e., a low missed detection probability.
The dimension of the uplink channel matrix for CE is reduced by considering only the UEs in $\Omega$.
Furthermore, the low-dimensional channel matrix $\left[{\bf W}_p\right]_{\Omega,:}$ is still sparse due to the angular-frequency structured sparsity of massive MIMO channels.
Hence, we can estimate $\left[{\bf W}_p\right]_{\Omega,:}, \forall p$, by applying the GMMV-AMP algorithm to (\ref{Eq:CS_Model_Vir_Act}), see \emph{line \ref{Step:CE}} of Algorithm \ref{Alg:Turbo-GMMV-AMP}.
Moreover, the signals received from the UEs in $\Gamma$, a subset of $\Xi^j$, are removed to enhance the sparsity of the uplink massive access channel matrix for AUD.
The residual received signals $\left(\widetilde {\bf Y}_p^G\right)^j$ are computed in \emph{lines \ref{Step:Re_Set}} and \emph{\ref{Step:Compute_Re}}, and are passed on to module A.

In the following turbo iterations ($j > 1$), the AUD problem in module A is to recover $\left({\bf X}_p^{\rm re}\right)^j$  based on the following model
\begin{equation}\label{Eq:CS_Model_Spa_Re}
\left(\widetilde {\bf Y}_p^G\right)^j = {{\bf S}_p^G}\left({\bf X}_p^{\rm re}\right)^j + {\bf N}_p,  \forall p \in \left[P\right],
\end{equation}
where $(\widetilde {\bf Y}_p^G)^j$ contains the residual received signals in the $j$-th turbo iteration, $\left({\bf X}_p^{\rm re}\right)^j = {\bf X}_p - {\widetilde {\bf X}_p^j}$, and ${\widetilde {\bf X}_p^j} \in \mathbb{C}^{K \times M}$ is defined as, $[{\widetilde {\bf X}_p}^j]_{\Gamma,:} = [\widehat {\bf X}_p^{j-1}]_{\Gamma,:}$, while $[{\widetilde {\bf X}_p^j}]_{{\cal K}-\Gamma,:} = {\bf 0}_{\left|{\cal K}-\Gamma\right|_c \times M}$.
To prevent the GMMV-AMP algorithm from diverging, we only remove the signals received from a part of the UEs in $\Xi^j$, i.e., $\lambda_{\rm aus} < 1$ (e.g. $\lambda_{\rm aus} = 0.8$).
Modules A and B will be executed iteratively.
Since the $\left({\bf X}_p^{\rm re}\right)^j$ become sparser and the channels of the UEs in $\Omega$ are iteratively re-estimated as the turbo iterations proceed, the $\Omega$ and the corresponding channels are constantly refined.
Therefore, compared to the simultaneous processing approaches in Section \ref{Sub:S_AUD_CE}, the proposed alternating approach facilitates more reliable AUD and CE with significantly smaller $G$, which means a dramatic reduction of access latency.

\begin{algorithm}[t]
\caption{Turbo-GMMV-AMP Algorithm}
\label{Alg:Turbo-GMMV-AMP}
\begin{algorithmic}[1]
\REQUIRE $\forall p$\;: Noisy observations ${\bf Y}_p^G$, pilot  matrices ${\bf S}_p^G$; the maximum number of turbo iterations $T_{\rm tur}$.
\ENSURE The estimated AUS $\widehat {\cal A}$ and the corresponding channel vectors $\{\widehat {\bf h}_{p,k}\}_{p=1}^P, \forall k \in {\widehat {\cal A}}$.
\STATE Initialization: $j\! =\! 1$, $\Xi^0\! =\! \emptyset$; $\left(\widetilde {\bf Y}_p^G\right)^1\! =\! {\bf Y}_p^G$.
\REPEAT
\STATE $k = 0$, $\Omega = \Gamma = \emptyset$.
\label{Step:Iter_Start}
\STATE // Module A: Active UE detector
\STATE $\forall p,k,m$\;: Acquire the belief indicators $\pi_{p,k,m}^j$ by applying the GMMV-AMP algorithm to model (\ref{Eq:CS_Model_Spa_Re}).\!\!\!\!\!\!
\label{Step:Module_A}
\FOR {$k \le K$}
\label{Step:Det_Start}
\IF {$\frac{1}{PM}\sum_p{\sum_m{r\left(\pi_{p,k,m}^j; \varepsilon_{\rm det}\right) \ge p_{\rm bi}}}$}
\STATE $\Omega  = \Omega \cup \Xi^{j-1} \cup \left\{k\right\}$.\\
\ENDIF
\IF {$\frac{1}{PM}\sum_p{\sum_m{r\left(\pi_{p,k,m}^j; \varepsilon_{\rm rel}\right) \ge p_{\rm bi}}}$}
\STATE $\Xi^j  = \Xi^{j-1}  \cup \left\{k\right\}$.\\
\ENDIF
\ENDFOR
\label{Step:Det_End}
\STATE // Module B: Virtual angular domain channel estimator
\vspace{-2mm}
\STATE $\forall p$\;: ${\bf R}_p^G = {{\bf Y}_p^G}{{\bf A}_R^*}$, ${\widehat {\bf W}_p}^j = {\bf 0}_{K \times M}$. \% ${\widehat {\bf W}_p}^j$ is the estimated channel matrix in the virtual angular domain.
\vspace{-2mm}
\STATE $\forall p$\;: Acquire the channel vectors $\left[{\bf W}_p^j\right]_{k,:}, \forall k\! \in\! \Omega$ by applying the GMMV-AMP algorithm to model (\ref{Eq:CS_Model_Vir_Act}).\!\!\!\!\!\!
\label{Step:CE}
\STATE Acquire set $\Gamma$, $\Gamma \subseteq \Xi^j$, and $\left|\Gamma\right|_c / \left|\Xi^j\right|_c = \lambda_{\rm aus}$. \% The elements in $\Gamma$ are randomly selected from $\Xi^j$.
\label{Step:Re_Set}
\STATE ${\widehat {\bf X}_p^j} = {\widehat {\bf W}_p^j}{{\bf A}_R^{\rm T}}$;\; $\left(\widetilde {\bf Y}_p^G\right)^{j+1} = {{\bf Y}_p^G - \left[{\bf S}_p^G\right]_{:,\Gamma}}\left[ \widehat {\bf X}_p^j\right]_{\Gamma,:}$.
\label{Step:Compute_Re}
\STATE $j = j + 1$.
\label{Step:Iter_End}
\UNTIL $j\! > \!T_{\rm tur}$.
\STATE ${\widehat {\cal A}} = \Omega$; ${\widehat {\bf h}_{p,k}} = \left(\left[\widehat {\bf W}_p^{j-1}\right]_{k,:}{{\bf A}_R^{\rm T}}\right)^{\rm T}$, $\forall p,k$.
\RETURN ${\widehat {\cal A}}$;\; $\{\widehat {\bf h}_{p,k}\}_{p=1}^P, \forall k \in {\widehat {\cal A}}$.
\end{algorithmic}
\end{algorithm}

\begin{algorithm}[t]
\caption{CS-Based Adaptive AUD and CE Scheme}
\label{Alg:Proposed_AUD_CE}
\begin{algorithmic}[1]
\STATE Determine the initial time slot overhead $G_0$ and set iteration index $i$ to 0.
\REPEAT
\STATE $\forall p$\;: Collect ${\bf Y}_p^{G_i}$ and ${\bf S}_p^{G_i}$ for given $G_i$. \% $G_i$ is the required $G$ for massive access in the $i$-th iteration.
\label{Step:Collect_Signals}
\STATE $\forall p$\;: Acquire the AUS estimate $\widehat {\cal A}$ and the corresponding CSI estimates $\widehat {\bf h}_{p,k}$, $k \in {\widehat {\cal A}}$ by leveraging Algorithm \ref{Alg:Turbo-GMMV-AMP}.
\STATE $\forall p$\;: Obtain $\widehat {\bf X}_p$ according to $\widehat {\cal A}$ and $\widehat {\bf h}_{p,k}$, $k \in \widehat {\cal A}$.
\label{Step:AUD_CE}
\STATE $\forall p$\;:\! \begin{small}${\bf Y}_p^{G_{i+1}}\!\! =\!\! \left[({\bf Y}_p^{G_i})^{\rm T}, {\bf y}_p^{G_{i+1}}\right]^{\rm T}$;\! ${\bf S}_p^{G_{i+1}}\!\! =\!\! \left[({\bf S}_p^{G_i})^{\rm T}, {\bf s}_p^{G_{i+1}}\right]^{\rm T}$.\end{small}\!\!
\STATE $G_{i+1} = G_i + 1$; $i = i + 1$.
\UNTIL $\sum_p{\left\|{\bf Y}_p^{G_{i-1}} - {\bf S}_p^{G_{i-1}}{\widehat {\bf X}_p}\right\|_{\rm F}^2} / \left(PG_{i-1}\right) < \epsilon$.
\label{Step:Evaluate}
\end{algorithmic}
\end{algorithm}

\emph{2) CS-Based Adaptive AUD and CE (Scheme 4):}
For practical systems, the UE activity and the channel environment are time-varying.
As a result, the sparsity level of the uplink massive access channel matrix may change over time.
If the channel matrix is relatively sparse, a small time slot overhead $G$ is sufficient to acquire accurate AUS and CSI estimates, while if the channel matrix is relatively dense, a large $G$ is required to guarantee reliable sparse signal recovery.
This motivates us to propose a CS-based adaptive AUD and CE scheme, as shown in Algorithm \ref{Alg:Proposed_AUD_CE}, to adaptively adjust $G$ to facilitate low-latency and high-reliability AUD and CE.
Algorithm \ref{Alg:Proposed_AUD_CE} can be summarized as follows.
\begin{itemize}
\item{\textbf{Step 1:} In each time slot, all active UEs transmit non-orthogonal RA pilots to the BS.
The pilots are pre-designed based on DCS theory and known to the system.}

\item{\textbf{Step 2:} The BS collects the received signals over multiple successive time slots, and the Turbo-GMMV-AMP algorithm is utilized to alternately acquire the AUS and CSI estimates, see \emph{lines \ref{Step:Collect_Signals}}-\emph{\ref{Step:AUD_CE}} of Algorithm \ref{Alg:Proposed_AUD_CE}.
Besides, the estimation reliability is evaluated based on a pre-specified criterion, see \emph{line \ref{Step:Evaluate}}} of Algorithm \ref{Alg:Proposed_AUD_CE}.

\item{\textbf{Step 3:} If the criterion is met, the outcome of the evaluation is informed to all UEs, and the active UEs stop transmitting pilots and start to transmit their data to the BS without scheduling permission.
Otherwise, \emph{Step 1} is repeated until the received signals collected at the BS are sufficient to meet the evaluation criterion.}
\end{itemize}
Here, based on CS theory, $G_0 \ge {\mathbb E}[|{\rm supp}\{\left[{\bf W}_p\right]_{:,m}\}|_c]$ is desirable, and $\epsilon = 0.8$ is suggested in [27, Sec. V].

\subsection{Computational Complexity Analysis}
\label{Sub:Complexity_Analysis}

For the practical implementation of the algorithms, the related computational complexity determines the hardware cost and the power consumption for processing.
Hence, the complexity analysis for the proposed algorithms is an important issue, especially for massive connectivity with large scale systems.

Table I compares the complexity of the proposed GMMV-AMP algorithm and Turbo-GMMV-AMP algorithm, as well as the conventional greedy CS recovery algorithms, i.e., generalized subspace pursuit (GSP) \cite{GSP}, simultaneous orthogonal matching pursuit (SOMP) \cite{SOMP}, and distributed sparsity adaptive matching pursuit DSAMP \cite{{Frequency common supp 1}, {Gao_CM'18}}, in terms of the number of the required complex multiplications in each iteration for AUD and CE.
Obviously the matrix inversion implemented in three greedy CS recovery algorithms for least square operation contributes to most of the computational complexity.
Hence, three greedy methods have the same order of computational complexity, i.e., in order of cubic magnitude of the number of active UEs $K_a$.
By contrast, the complexity of the proposed GMMV-AMP and Turbo-GMMV-AMP algorithms increases linearly with $K, G, M$, and $P$.
Hence, for massive access scenarios with large $K_a$, the proposed algorithms can be more computationally efficient.

\begin{table}[t]
\caption{Computational Complexity for AUD and CE}
\label{Tab:Scheme_Compare}
\vspace{-0mm}
\begin{threeparttable}
\begin{center}
\scalebox{0.74}{
\begin{tabular}{cc}
\hline
\hline
Algorithm& Number of complex multiplications in each iteration\\
\hline
GSP& $2(G+1)KMP + GM^2P + 2(M+1)GPK_a^2 + 2PK_a^3$\\
SOMP& $(2G+1)KMP + GM^2P + (M+1)GPo^2 + Po^3$\\
DSAMP& $(2G+3)KMP + MP + 2(M+1)GPu^2 + 2Pu^3$\\
GMMV-AMP& $4GKMP + 3GKP + 16GMP + 20KMP$\\
Turbo-GMMV-AMP& $2T_{\rm amp}(4GKMP + 3GKP + 16GMP + 20KMP)$\\
\hline
\hline
\end{tabular}}
\end{center}
\begin{tablenotes}
\footnotesize
\item Note: $o$ is the iteration index \cite{SOMP}, and $u$ denotes the stage index \cite{{Frequency common supp 1}}.
\end{tablenotes}
\end{threeparttable}
\vspace{-4.5mm}
\end{table}

\section{State Evolution}

SE is a framework for analyzing the performance of AMP algorithms in the large system limit where $K \rightarrow \infty$ \cite{{AMP_SE}}.
In this section, we harness SE to characterize the mean square error (MSE) performance of the proposed GMMV-AMP algorithm.
The MSE of the estimation and the variance of the estimated channels are defined as
\begin{align}
e^q &= \frac{1}{KMP}\sum\nolimits_p{\sum\nolimits_k{\sum\nolimits_m{ \left|{\hat x}_{p,k,m}^q - x_{p,k,m}\right|^2 }}}, \label{Eq:SE_MSE}\\
\vartheta^q &= \frac{1}{KMP}\sum\nolimits_p{\sum\nolimits_k{\sum\nolimits_m{ v_{p,k,m}^q }}}, \label{Eq:SE_Var}
\end{align}
respectively.
Based on the derivations in Appendix \ref{Proof:Decouple Post}, the GMMV-AMP algorithm can be explained intuitively.
For the $p$-th pilot subcarrier, the proposed GMMV-AMP algorithm decouples the matrix estimation problem in (\ref{Eq:CS_Model_Spa}) into $KM$ independent scalar estimation problems, as
\begin{equation}\label{Eq:AMP_Dcouple}
{\bf Y} = {\bf SX} + {\bf N} \to C_{k,m}^q = x_{k,m} + {\tilde n_{k,m}^q},\; \forall k,m,
\end{equation}
where index $p$ and superscript $G$ are omitted for notational simplicity, $C_{k,m}^q \sim {\cal CN}\left(C_{k,m}^q; x_{k,m}, D_{k,m}^q\right)$ is the equivalent measurement of $x_{k,m}$ in the $q$-th iteration, and ${\tilde n_{k,m}^q} \sim {\cal CN}\left({\tilde n_{k,m}^q}; 0, D_{k,m}^q\right)$ denotes the effective noise.

\begin{prop}
Define a scalar random variable ${X_0}\! \sim\! p_0\left(X\right)$.
Then, we have $C^q = x_0 + {\tilde n^q}$, ${\tilde n^q} \sim {\cal CN}\left({\tilde n^q}; 0, D^q\right)$, and the posterior distribution of $x_0$ can be expressed as
\begin{equation}
p\left(x_0|C^q,D^q\right) = \frac{1}{\widetilde Z_3}p_0\left(x_0\right){\cal CN}\left(x_0;C^q,D^q\right),
\end{equation}
where
\begin{equation}\label{Eq:SE_C_D}
C^q = x_0 + \sqrt{\frac{\sigma + Ke^q}{G}}z,\; D^q = \frac{\sigma + K\vartheta^q}{G},
\end{equation}
and $z\! \sim\! {\cal CN}\left(z;0,1\right)$.
Hence, $e^{q+1}$ and $\vartheta^{q+1}$ are updated as
\begin{align}
e^{q+1} &= \int{\int{ \left|g_a\left(C^q,D^q\right) - x_0\right|^2}}{\cal D}x_0{\cal D}z, \label{Eq:SE_MSE_Update} \\
\vartheta^{q+1} &= \int{\int{ g_c\left(C^q,D^q\right)}}{\cal D}x_0{\cal D}z, \label{Eq:SE_Var_Update}
\end{align}
where ${\cal D}x_0 = p_0\left(x_0\right)d{x_0}$ and ${\cal D}z = e^{-\left|z\right|^2}/\pi dz$.
Therefore, defining a scalar random variable following the same distribution as the channel coefficient, i.e., $X_0 \sim p\left(X\right)$, $C_{k,m}^q$ and $D_{k,m}^q$ in the GMMV-AMP algorithm can be calculated as in (\ref{Eq:SE_C_D}), and $e^{q+1}$ and $\vartheta^{q+1}$ can be obtained from (\ref{Eq:SE_MSE_Update}) and (\ref{Eq:SE_Var_Update}).
\end{prop}

\begin{proof}
Please refer to the Appendix \ref{Proof:SE}.
\end{proof}

Since the a priori distribution $p_0\left(X\right)$ does not take the structured sparsity of channel matrix into consideration, the scalar SE in (\ref{Eq:SE_C_D})-(\ref{Eq:SE_Var_Update}) can not accurately analyze the MSE performance of the proposed GMMV-AMP algorithm.
Hence, we use Monte Carlo simulation to carry out the SE, so that (\ref{Eq:SE_MSE_Update}) and (\ref{Eq:SE_Var_Update}) are simplified and the exploitation of the structured sparsity is also taken into into account.
In contrast to the conventional AMP algorithms in \cite{{AMP-based MRA 1}} and \cite{{AMP-based MRA 2}}, which assume full knowledge of the a priori distribution of the channels and the noise variance, the SE for the proposed GMMV-AMP algorithm also needs to track the update rules of the hyper-parameters ${\bm{\theta }}$ in (\ref{Eq:Pri_Mean})-(\ref{Eq:Pri_Var}), and \cite{A_R}
\begin{equation}\label{Eq:SE_Noi_Update}
\sigma_{k,m}^{q+1} = \frac{\sigma_0 + e^q}{\left[1 + \vartheta^q/\sigma_{k,m}^q\right]} + \frac{\sigma_{k,m}^q\vartheta^q}{\sigma_{k,m}^q + \vartheta^q},
\end{equation}
where $\sigma_0$ is the actual noise variance.
The SE of the proposed GMMV-AMP algorithm is summarized in Algorithm \ref{Alg:SE}.

\begin{algorithm}[t]
\caption{State Evolution of GMMV-AMP Algorithm}
\label{Alg:SE}
\begin{algorithmic}[1]
\REQUIRE $\gamma = K_a/K$, $\kappa = G/K$, $M$, and $P$; $T_{\rm amp}$, $\rho$, and $\eta$.
\ENSURE The predicted MSE performance, ${\rm MSE}_{\rm SE}$.
\STATE Determine the number of Monte Carlo realizations $\widetilde K$, then ${\widetilde K_a} = \gamma{\widetilde K}$ and ${\widetilde G} = \kappa{\widetilde K}$.
\STATE Generate the Monte Carlo samples $\left\{{\bf X}_p\right\}_{p=1}^P$ according to the a priori distribution ${p_0}\left({\bf X}\right)$ and the sparsity structure.
\STATE $\forall p,k,m$: Set the iteration index $q$ to 1, initialize $e^q$ and $\vartheta^q$ as $e^1=\vartheta^1=1$, and initialize the hyper-parameters as $\gamma_{p,k,m}^1 = \gamma$, $\sigma_{p,k,m}^1 = 1$, $\mu_{p,k,m}^1 = 0$, $\tau_{p,k,m}^1 = 1$.
\REPEAT
\STATE $\forall p,k,m:C_{p,k,m}^q = x_{p,k,m} + \sqrt{\frac{\sigma_{p,k,m}^q + {\widetilde K}e^q}{\widetilde G}}z.$
\STATE $\forall p,k,m:C_{p,k,m}^q = {\rho}C_{p,k,m}^{q-1} + \left(1-\rho\right)C_{p,k,m}^q.$
\STATE $\forall p,k,m:D _{p,k,m}^q = {\rho}D_{p,k,m}^{q-1} + \left(1-\rho\right)\frac{\sigma_{p,k,m}^q + {\widetilde K}\vartheta^q}{\widetilde G}.$
\STATE $e^{q+1}\! =\! \frac{1}{KMP}\!\! \sum\nolimits_p\!{\sum\nolimits_k\!{\sum\nolimits_m\!{\left|g_a(C_{p,k,m}^q,D _{p,k,m}^q)\! -\! x_{p,k,m}\right|^2}}}.$
\STATE $\vartheta^{q+1} = \frac{1}{KMP}\sum\nolimits_p{\sum\nolimits_k{\sum\nolimits_m{g_c(C_{p,k,m}^q,D_{p,k,m}^q)}}}.$
\STATE Update the hyper-parameters via (\ref{Eq:Pri_Mean})-(\ref{Eq:Pri_Var}) and (\ref{Eq:SE_Noi_Update}).
\STATE Refine update rule of $\gamma_{p,k,m}$ as \emph{line \ref{Step:Hyper_Refine}} of Algorithm \ref{Alg:GMMV-AMP}.
\STATE $q = q + 1$.
\UNTIL $q > T_{\rm amp}$ or $\left|e^q - e^{q-1}\right| < \eta$.
\RETURN ${\rm MSE}_{\rm SE} = e^q$.
\end{algorithmic}
\end{algorithm}

\section{Simulation Results}
\label{Sec:Sim Result}

For the presented simulation results, we assume the BS employs a ULA with $M$ antennas, $K = 500$ potential UEs are randomly distributed in the cell with radius 1 km, and $K_a = 50$ ($K_a \ll K$) UEs are active unless otherwise specified.
Furthermore, the carrier frequency is 2 ${\rm GHz}$, the bandwidth is $B_s = 10$ ${\rm MHz}$, and the received SNR is 30 ${\rm dB}$.
The system adopts OFDM for massive access in an eMBB scenario, where $N = 2048$ subcarriers and a cyclic prefix of length $N_{\rm CP}=64$ are employed.
$P = N_{\rm CP}$ pilots are uniformly allocated to the $N$ subcarriers \cite{{Frequency common supp 1}}.
We consider the one-ring channel model with limited angular spread \cite{One-ring}, so that each UE's massive MIMO channel exhibits clustered sparsity in the virtual angular domain.
The large scale fading $\rho_k$ is given by the standard Log-distance path loss model as $\rho_k = 128.1 + 37.6{\rm lg}(d_k)$ with distance $d_k$ measured in km.
The small scale fading channel is generated by (\ref{Eq:Channel_Model}), where $L$ varies from 8 to 40 \cite{A_R}, the related AOAs are generated within an angular spread $\Delta$ varying from \ang{20} to \ang{40} so that the effective sparsity level in virtual angle domain $S_a$ varies from 8 to 14,  ${\beta_{k,l}} \sim {\cal CN}({\beta_{k,l}}; 0, 1)$, and $\varpi_{k,l}$ is randomly and uniformly selected from $[0, N_{CP}/B_s]$.
Furthermore, $T_{\rm amp} = 200$, $T_{\rm tur} = 10$, $\eta = 10^{-5}$, and the simulation results are obtained by averaging over $N_{\rm sim} = 3 \times 10^3$ simulation runs unless otherwise specified.

\begin{table}[t]
\scriptsize
\caption{Considered Schemes}
\label{Tab:Scheme_Compare}
\vspace{-1.5mm}
\begin{center}
\begin{tabular}{|c|c|c|c|c|}
\hline
\multicolumn{2}{|c|}{Scheme}                                                                  &Model          &Algorithm        & $G$        \\
\hline
\multirow{2}{*}{\begin{tabular}[c]{@{}c@{}}Simultaneous\\ AUD and CE\end{tabular}} & Scheme 1 & (\ref{Eq:CS_Model_Spa})          & GMMV-AMP        & fixed    \\ \cline{2-5}
                                                                                   & Scheme 2 & (\ref{eq:CS Model Vir})          & GMMV-AMP        & fixed    \\
\hline
\multirow{2}{*}{\begin{tabular}[c]{@{}c@{}}Alternating\\ AUD and CE\end{tabular}}  & Scheme 3 & (\ref{Eq:CS_Model_Spa})  and (\ref{eq:CS Model Vir}) & Turbo-GMMV-AMP  & fixed  \\ \cline{2-5}
                                                                                   & Scheme 4 & (\ref{Eq:CS_Model_Spa})  and (\ref{eq:CS Model Vir}) & Turbo-GMMV-AMP  & adaptive \\
\hline
\multicolumn{2}{|c|}{\begin{tabular}[c]{@{}c@{}}Baseline\\ Schemes\end{tabular}}              & (\ref{Eq:CS_Model_Spa})          & \begin{tabular}[c]{@{}c@{}}GSP {[}32{]}, SOMP {[}31{]}, \\ and DSAMP {[}21{]}\end{tabular}                                                                                                      & fixed\\
\hline
\end{tabular}
\end{center}
\end{table}

An overview of the proposed \emph{Schemes 1-4} and the considered baseline schemes, GSP \cite{GSP}, SOMP \cite{SOMP}, and DSAMP \cite{{Frequency common supp 1}}, is presented in Table \ref{Tab:Scheme_Compare}.
For performance evaluation, we consider the detection error probability $P_e$ for AUD and the MSE for CE, which are respectively defined as
\begin{equation}
P_e = \frac{\sum_k{\left|{\widehat \alpha_k} - {\alpha _k}\right|}}{K},\   {\rm MSE} =  \frac{\sum_p{\left\|{\widehat {\bf X}_p} - {\bf X}_p\right\|_{\rm F}^2}}{KMP}.
\end{equation}
Here, to reduce computational complexity, for \emph{Schemes 1-4}, only $\widetilde P$ out of $P$ pilot subcarriers are used to estimate AUS $\cal A$ and the corresponding channels.
The remaining $(P - \widetilde P)$ pilot subchannels of the active UEs can be easily estimated by applying the GMMV-AMP algorithm to (\ref{Eq:CS_Model_Vir_Act}) given $\Omega = {\widehat {\cal A}}$.

\begin{figure}[t]
	\centering
	\includegraphics[width=1\columnwidth,keepaspectratio]
    {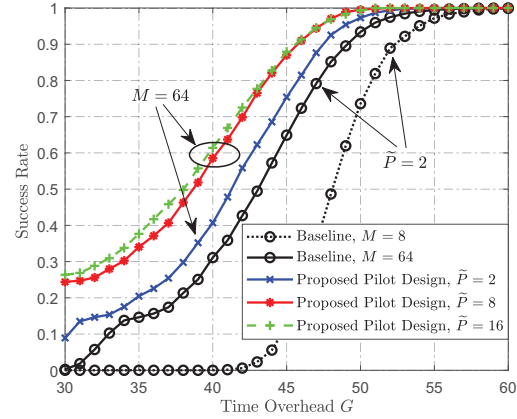}
	\caption{Comparison of success rate for the proposed DCS-based pilot design and identical pilots (Baseline).}
    \label{Fig:Pilot Design}
    \vspace{-2mm}
\end{figure}
Fig. \ref{Fig:Pilot Design} verifies the superiority of the proposed DCS-based pilot design for broadband massive access based on \emph{Scheme 1}.
The quality of the pilots is evaluated in terms of the success rate, which is defined as the ratio of the number of simulation runs with $P_e = 0$ to the total number of simulation runs.
Fig. \ref{Fig:Pilot Design} shows that, as expected, employing different ${\bf S}_p^G$ for different $p$ improves the success rate.
Moreover, the performance is further improved when massive MIMO and larger $\widetilde P$ are employed, since the structured sparsity of $\left\{{\bf X}_p\right\}_{p=1}^P$ is leveraged.

\begin{figure}[t]
	\centering
	\includegraphics[width=1\columnwidth,keepaspectratio]
    {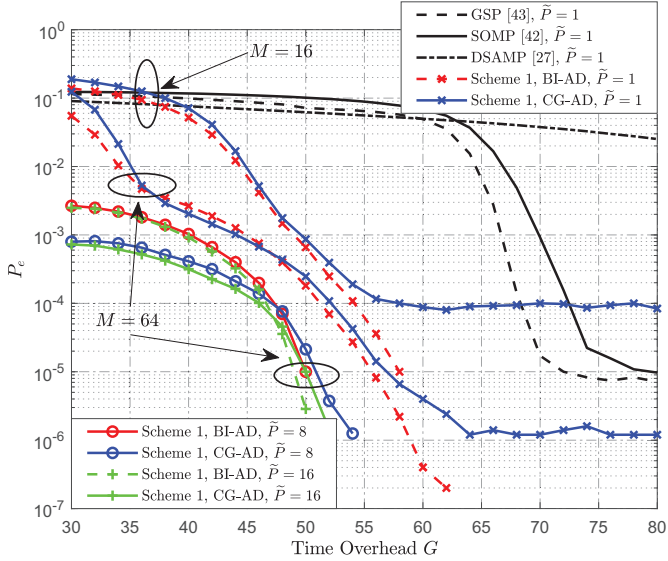}
	\caption{Comparison of detection error probabilities for \emph{Scheme 1} and state-of-the-art GMMV-CS algorithms.}
    \label{Fig:Pe_S1_Class}
    \vspace{-2mm}
\end{figure}

Fig. \ref{Fig:Pe_S1_Class} examines the AUD performance of \emph{Scheme 1}, where the performances of three state-of-the-art GMMV-CS algorithms are shown as benchmarks.
As can be observed, the GMMV-AMP algorithm outperforms the other three algorithms.
Hence, the access latency can be considerably reduced for a given target $P_e$.
For example, for $M = 16$ and ${\widetilde P} = 1$, the GSP-based scheme requires $G = 72$ to achieve $P_e \le 10^{-5}$,  whereas \emph{Scheme 1} with BI-AD needs only $G = 58$, which indicates a reduction of approximately $19\%$ in the access latency.
Moreover, \emph{Scheme 1} can achieve a better detection performance by equipping more antennas at the BS and/or utilizing larger $\widetilde P$, since a larger $M$ and/or $\widetilde P$ can enhance the space-frequency structured sparsity shown in Fig. \ref{Fig:Struct_Sparse}(a).
However, this improvement becomes negligible when $M$ and $\widetilde P$ are sufficiently large.

In Fig. \ref{Fig:Pe_S1_Class}, we further compare the performance of CG-AD and BI-AD.
For $G < 55$ or ${\widetilde P} \ge 8$, CG-AD and BI-AD have similar performance.
However, when $G \ge 56$ and $\widetilde P \le 8$, BI-AD outperforms CG-AD, as CG-AD suffers from a detection error floor.
The reason for this behavior is that when there are sufficiently many measurements for reliable CS recovery, the belief indicator $\pi_{p,k,m}$ takes values of 0 and 1, but the estimated channel gain ${\hat x_{p,k,m}}$ takes the true value of $x_{p,k,m}$.
Hence, based on the threshold function $r\left(x; \varepsilon\right)$, BI-AD can reliably determine whether $x_{p,k,m}$ is zero or not, while for CG-AD, missed detections and false alarms can not be avoided, which may lead to a detection error floor.
Clearly, for \emph{Scheme 1}, BI-AD is more reliable than CG-AD for AUD.

\begin{figure}[t]
	\centering
	\includegraphics[width=1\columnwidth,keepaspectratio]
    {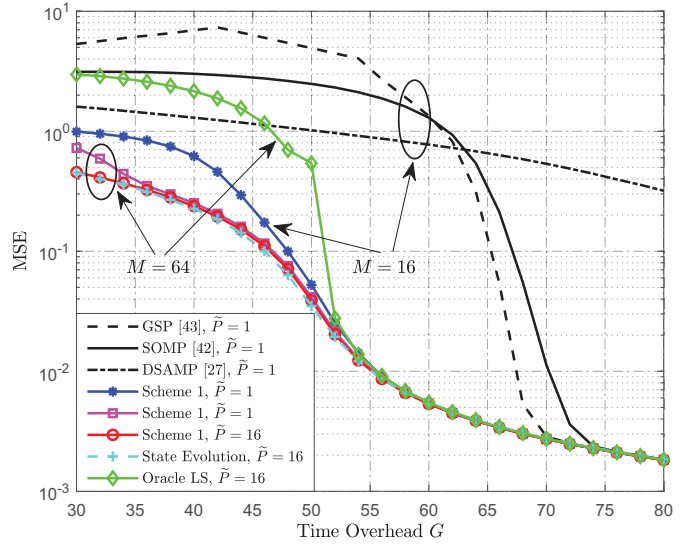}
	\caption{MSE performance of \emph{Scheme 1}, state-of-the-art GMMV-CS algorithms, oracle LS, and SE.}
    \label{Fig:MSE_S1_Class}
\end{figure}

Fig. \ref{Fig:MSE_S1_Class} depicts the CE MSE performance of \emph{Scheme 1} and the three state-of-the-art GMMV-CS algorithms also considered in Fig. \ref{Fig:Pe_S1_Class}.
The oracle least square (LS) estimator with known AUS $\cal A$ is used as performance upper bound \cite{{Frequency common supp 1}}.
When $G$ is sufficiently large, both \emph{Scheme 1} and the three baseline schemes approach the oracle LS performance bound, since ${\cal A}$ is accurately estimated in this case, and the CE problem is reduced to an oracle LS problem.
However, for $G < 70$, \emph{Scheme 1} outperforms the three baseline algorithms, and its performance improves when $M$ and/or $\widetilde P$ increase.
Besides, the MSE performance of \emph{Scheme 1} is accurately predicted by SE.
Here, an important observation is that when $G < K_a$, both \emph{Scheme 1} and the oracle LS estimator can not perform reliable CE.
This suggests that $G \ge K_a$ is required for reliable CE in (\ref{Eq:CS_Model_Spa}).
Hence, the reduction of $G$ is limited to $K_a$, which is identical to the sparsity level of the column vectors in ${\bf X}_p$.
This motivates \emph{Scheme 2} for simultaneous AUD and CE, where the sparsity level of the channel matrix, defined as ${\widetilde K_a} = {\rm max} \left\{|{\rm supp}\{\left[{\bf W}_p\right]_{:,m}\}|_c, \forall p,m \right\}$, is less than $K_a$.

\begin{figure}[t]
	\centering
	\includegraphics[width=1\columnwidth,keepaspectratio]
    {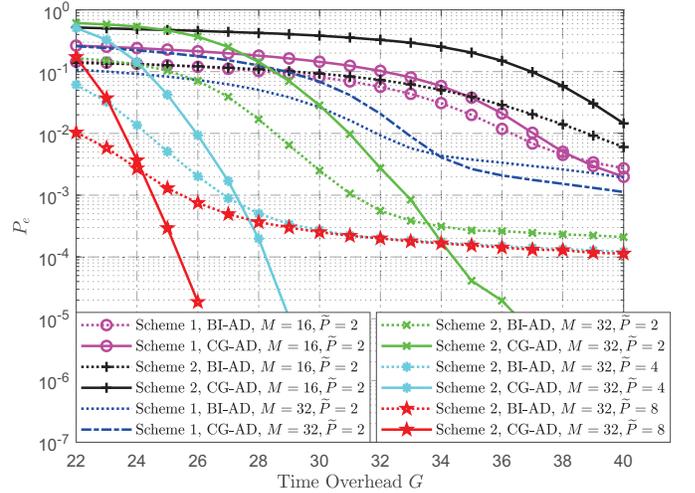}
	\caption{Comparison of detection error probabilities for \emph{Scheme 1} and \emph{Scheme 2}.}
    \label{Fig:Pe_S1_S2}
    \vspace{-2mm}
\end{figure}

Fig. \ref{Fig:Pe_S1_S2} compares the AUD performance of \emph{Scheme 1} and \emph{Scheme 2}.
For $M = 16$, \emph{Scheme 1} outperforms \emph{Scheme 2} for both BI-AD and CG-AD, respectively.
In contrast, for $M  = 32$, \emph{Scheme 2} achieves a much better AUD performance than \emph{Scheme 1} when $G > 28$.
This is because ${\bf W}_p$ is sparser than ${\bf X}_p$, i.e., ${\widetilde K_a} < K_a$, and the required $G$ for reliable AUD and CE in \emph{Scheme 2} mainly depends on ${\widetilde K_a}$  rather than $K_a$.
However, the virtual angular domain sparsity weakens the common sparsity pattern across multiple columns of ${\bf W}_p$.
Therefore, if the BS has a small number of antennas, e.g., $M = 16$, \emph{Scheme 1} outperforms \emph{Scheme 2} by leveraging the common sparsity observed at different BS antennas.
However, when $M$ becomes large, \emph{Scheme 2} can considerably reduce the required $G$ for reliable AUD and CE compared to \emph{Scheme 1}.
In addition, BI-AD in \emph{Scheme 2} suffers from an obvious detection error floor, which indicates that for \emph{Scheme 2}, CG-AD is more reliable than BI-AD for detection of the active UEs.
Fig. \ref{Fig:MSE_S1_S2} compares the CE MSE performance of \emph{Scheme 1} and \emph{Scheme 2}, which again verifies the superiority of \emph{Scheme 2} for massive MIMO systems.
The theoretical SE accurately predicts the MSE.

\begin{figure}[t]
	\centering
	\includegraphics[width=1\columnwidth,keepaspectratio]
    {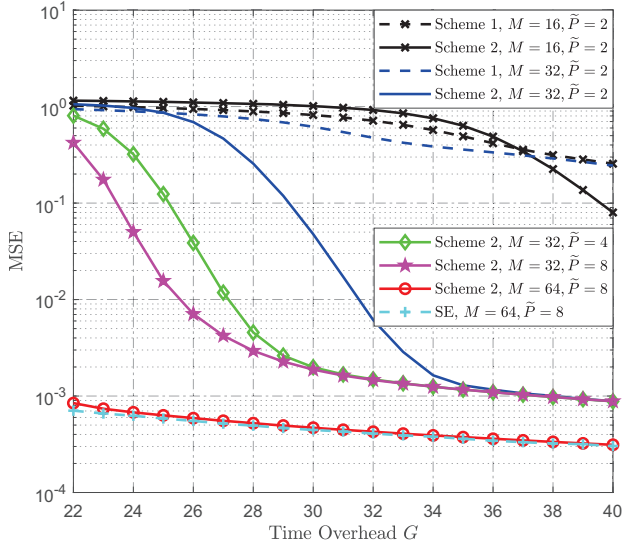}
	\caption{MSE performance of \emph{Scheme 1}, \emph{Scheme 2}, and SE.}
    \label{Fig:MSE_S1_S2}
\end{figure}

\begin{figure}[t]
	\centering
	\includegraphics[width=1\columnwidth,keepaspectratio]
    {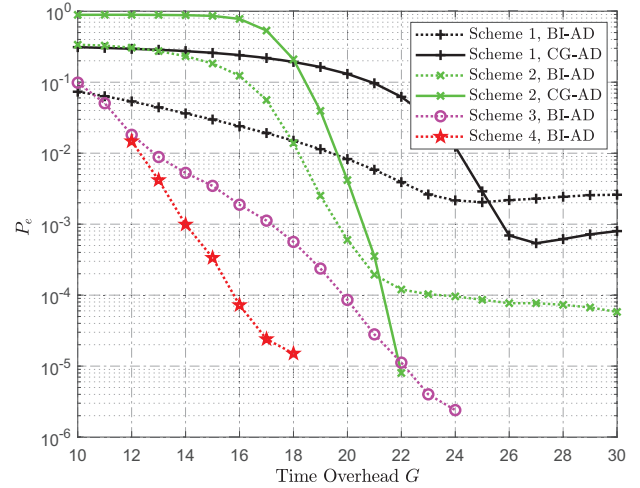}
	\caption{Comparison of detection error probabilities of \emph{Schemes 1-4} for $M = 64$, ${\widetilde P} = 8$, and $N_{\rm sim} = 2000$.}
    \label{Fig:Pe_S1-S4}
    \vspace{-2mm}
\end{figure}

\begin{figure}[t]
	\centering
	\includegraphics[width=1\columnwidth,keepaspectratio]
    {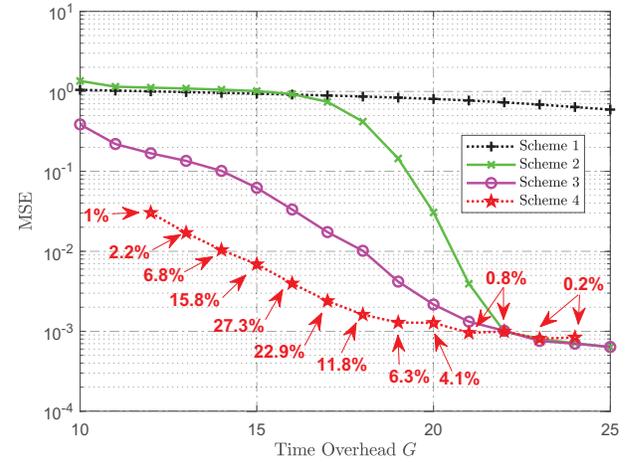}
	\caption{MSE performance of \emph{Schemes 1-4}, where $M = 64$, ${\widetilde P} = 8$, $N_{\rm{sim}} = 2000$, and the percentage of simulation runs requiring a given consumed time slot overhead $G$ is indicated.}
    \label{Fig:MSE_S1-S4}
\end{figure}

\begin{figure}[t]
	\centering
	\includegraphics[width=1\columnwidth,keepaspectratio]
    {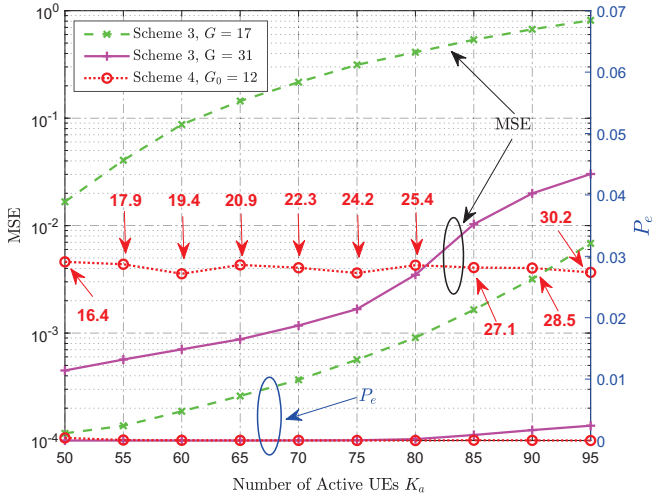}
	\caption{{MSE and $P_e$ performance of \emph{Scheme 3} and \emph{Scheme 4} for different $K_a$, where $M = 64$, ${\widetilde P} = 8$, $N_{\rm{sim}} = 1000$, and the average time slot overhead of \emph{Scheme 4} is indicated.}}
    \label{Fig:Pe_MSE_S3_S4}
    \vspace{-2mm}
\end{figure}

Fig. \ref{Fig:Pe_S1-S4} and Fig. \ref{Fig:MSE_S1-S4} compare the MSE and $P_e$ performance of \emph{Schemes 1-4}, respectively.
The details of \emph{Scheme 4} are shown in Algorithm \ref{Alg:Proposed_AUD_CE}.
Given $\gamma = K_a/K = 0.1$ and $M = 64$, an initial time slot overhead of
\begin{equation}\label{Eq:G_0}
\begin{aligned}
G_0 &= 1.5{\mathbb E}[|{\rm supp}\{\left[{\bf W}_p\right]_{:,m}\}|_c]\\
    &= 1.5{\gamma}K{\mathbb E}[S_a]/M \approx 12
\end{aligned}
\end{equation}
is adopted.
For different simulation runs, the varying $S_a$ yields a different sparsity level for $\left\{{\bf W}_p\right\}_{p=1}^P$, i.e., different ${\widetilde K_a}$.
The results and the consumed $G$, i.e., the overhead, are recorded after the pre-defined criterion in \emph{line 8} of Algorithm \ref{Alg:Proposed_AUD_CE} is met.
For very low time slot overheads, i.e., $G \in \left[10,18\right]$, both \emph{Scheme 1} and \emph{Scheme 2} have a poor performance.
This is because for these schemes, $G \ll K_a$ leads to extremely insufficient measurements.
By contrast, \emph{Scheme 3} using the proposed Turbo-GMMV-AMP algorithm can achieve a much better AUD and CE performance than \emph{Scheme 1} and \emph{Scheme 2}, which confirms its superiority in reducing access latency.
Finally, \emph{Scheme 4}, i.e., the proposed CS-based adaptive AUD and CE scheme, adaptively adjusts the overhead to achieve satisfactory AUD and CE performance.
In Fig. \ref{Fig:MSE_S1-S4}, for \emph{Scheme 4}, the percentage of simulation runs requiring the given time slot overhead $G$ is provided.
This reveals that $88.6\%$ of the simulation runs require overheads of $G \in \left[12,18\right]$, and the corresponding MSE performance is much better than those for \emph{Schemes 1-3}.

Fig. \ref{Fig:Pe_MSE_S3_S4} further compares the AUD and CE performance of \emph{Scheme 3} and \emph{Scheme 4} for different number of active UEs $K_a$.
The proposed \emph{Scheme 4} adaptively adjusts the time slot overhead to guarantee reliable AUD and CE for different $K_a$.
However, \emph{Scheme 3} employs a fixed overhead and suffers from poor performance when $K_a$ becomes large.
This means that some of the active UEs will not be able to access the network.
Hence, for practical systems with time-varying UE activity, the superiority of the proposed CS-based adaptive AUD and CE scheme is evident.

\section{Conclusion}

This paper investigates new methods for facilitating massive access in massive MIMO systems, which leverage the sporadic traffic of the UEs and the virtual angular domain sparsity of massive MIMO channels to dramatically reduce the access latency.
The space-frequency structured sparsity of the channel matrix in the spatial domain improves the AUD performance, while the angular-frequency structured sparsity of the channel matrix in the angular domain improves the CE performance.
Therefore, joint AUD and CE schemes exploiting only spatial domain or only angular domain channel model can not take full advantage of the sparsity properties of massive access in massive MIMO systems.
This motivates the derivation of the proposed Turbo-GMMV-AMP algorithm, which achieves a significant performance improvement by performing AUD based on a spatial domain channel model and CE based on an angular domain channel model in an alternating manner.
Furthermore, for practical systems, where the number of active UEs is not known, the proposed CS-based adaptive AUD and CE scheme can adjust the time slot overhead to realize ultra-reliable low-latency massive access.

\appendices
\section{Proof of the Proposition 1}
\label{Proof:Decouple Post}

The factorization in the joint posterior probability (\ref{Eq:Joint_Post}) can be represented by a bipartite graph, which motivates the application of the sum-product (SP) algorithm to realize the MMSE estimator \cite{{Sum-Product}}.
As the bipartite graph consists of $M$ independent subgraphs, we only discuss the $m$-th subgraph in the following derivations, and the antenna index $m$ is dropped for notational simplicity.
For the $m$-th subgraph, we define variable nodes $V = \left[K\right]$, factor nodes of the likelihood function $F = \left[G\right]$, and edges $E = \left[K\right] \times \left[G\right] = \left\{\left(k,g\right): k \in \left[K\right], g \in \left[G\right]\right\}$.
The update rules for the messages associated to the edges are \cite{{AMP}}
\begin{align}
{\xi_{k \to g}^{q+1}}\left(x_k\right) &\propto p_0\left(x_k\right)\prod\limits_{b \ne g} {{\hat \xi_{b \to k}^q}\left(x_k\right)}, \label{Eq:Message_Var} \\
{\hat \xi_{g \to k}^q}\left(x_k\right) &\propto \int{\prod\limits_{j \ne k} {\xi_{j \to g}^q\left(x_j\right)p\left(y_g|{\bf x}\right)}}d{\bf x}_{\backslash k}, \label{Eq:Message_Fac}
\end{align}
where $b \in \left[G\right]$, $j \in \left[K\right]$, $q$ denotes the $q$-th iteration, and $\propto$ denotes equality up to a constant scale factor.

One practical hurdle for the large-scale implementation of the SP algorithm lies in the required evaluation of high-dimensional integrals for calculation of messages ${\hat \xi_{g \to k}^q}\left(x_k \right)$.
This leads to an unacceptably high complexity.
However, a key observation is that, in the large system limit with $K \to \infty$, the messages ${\hat \xi_{g \to k}^q}\left(x_k\right)$ can be approximated by Gaussian distributions \cite{{AMP}}.
Since the random variables $\left\{x_j\right\}_{j=1}^K$ are independently complex Gaussian distributed, random variable $Z_{g,k} = \sum\nolimits_{j \ne k} s_{g,j}x_j$ follows a complex Gaussian distribution $Z_{g,k} \sim {\cal CN}\left(z_{g,k}; Z_{g \to k}^q, V_{g \to k}^q\right)$, with mean $Z_{g \to k}^q$ and variance $V_{g \to k}^q$ given as
\begin{equation}\label{Eq:Z_Mean}
\begin{array}{l}
\!\!\! Z_{g \to k}^q = \sum\nolimits_{j \ne k} {s_{g,j}{\hat x_{j \to g}^q}},\; V_{g \to k}^q = \sum\nolimits_{j \ne k} {\left|s_{g,j}\right|^2v_{j \to g}^q}, \!\!\!
\end{array}
\end{equation}
where ${\hat x_{j \to g}^q}$ and $v_{j \to g}^q$ are the mean and variance of message ${\xi_{j \to g}^q}\left(x_j\right)$, respectively.
Hence, the messages ${\hat \xi_{g \to k}^q}\left(x_k\right)$ can be approximated as
\begin{equation}\label{Eq:Fac_Message_Approx}
\begin{array}{l}
{\hat \xi_{g \to k}^q}\left(x_k\right) \propto {\cal CN}\left(x_k; \frac{y_g - Z_{g \to k}^q}{s_{g,k}}, \frac{\sigma + V_{g \to k}^q}{\left|s_{g,k}\right|^2}\right).
\end{array}
\end{equation}
Further, the posterior distribution of $x_k$ is calculated as
\begin{equation}\label{Eq:Message_Mar_Post}
\begin{aligned}
{\xi_k^{q+1}}\left(x_k\right) &=  p_0\left(x_k\right)\prod\limits_g {{\hat \xi_{g \to k}^q}\left(x_k\right)}  \\
                              &\propto p_0\left(x_k\right){\cal CN}\left(x_k; C_k^q, D_k^q\right),
\end{aligned}
\end{equation}
where
\begin{equation}\label{Eq:D and C}
\begin{array}{l}
\!\!\! D_k^q =\! \left[\sum\nolimits_g {\frac{\left|s_{g,k}\right|^2}{\sigma + V_{g \to k}^q}}\right]^{-1},\; C_k^q = D_k^q\sum\nolimits_g {\frac{s_{g,k}^*\left(y_g - Z_{g \to k}^q \right)}{\sigma + V_{g \to k}^q}}. \!\!\!
\end{array}
\end{equation}
It is convenient to introduce a family of densities
\begin{equation}\label{Eq:Post}
p\left(x; C, D\right) = \frac{1}{\widetilde Z_4}p_0\left(x\right){\cal CN}\left(x; C, D\right),
\end{equation}
where ${\widetilde Z_4} = \int {p_0\left(x\right){\cal CN}\left(x; C, D\right)}dx$ is a normalization constant.
The corresponding mean and variance are
\begin{align}
g_a\left(C, D\right) &= \int {xp\left(x; C, D\right)}dx, \label{Eq:Define_Post_Mean} \\
g_c\left(C, D\right) &= \int {\left|x - g_a\left(C, D\right)\right|^2 p\left(x; C, D\right)}dx, \label{Eq:Define_Post_Var}
\end{align}
respectively.
Define ${\xi_{k}^{q+1}}(x_k) \propto {\cal CN}\left(x_k; {\hat x_k^{q+1}}, v_k^{q+1}\right)$, with ${\hat x_k^{q+1}} = g_a\left(C_k^q, D_k^q\right)$ and $v_k^{q+1} = g_c\left(C_k^q, D_k^q\right)$.
The messages ${\xi_{k \to g}^{q+1}}\left(x_k\right)$ can be approximated as
\begin{equation}\label{Eq:Var_Message_Approx}
{\xi_{k \to g}^{q+1}}\left(x_k\right) = \frac{{\xi_k^{q+1}}\left(x_k\right)}{\xi_{g \to k}^q\left(x_k\right)} \propto {\cal CN}\left(x_k; {\hat x_{k \to g}^{q+1}}, v_{k \to g}^{q+1} \right),
\end{equation}
where
\begin{align}
\frac{1}{v_{k \to g}^{q+1}} &= \frac{1}{v_k^{q+1}} - \frac{\left|s_{g,k}\right|^2}{\sigma + V_{g \to k}^q}, \\
\frac{1}{{\hat x_{k \to g}^{q+1}}} &= v_{k \to g}^{q+1}\left[\frac{\hat x_k^{q+1}}{v_k^{q+1}} - \frac{s_{g,k}^*\left(y_g - Z_{g \to k}^q\right)}{\sigma + V_{g \to k}^q}\right].
\end{align}
At this point, the messages ${\xi_{k \to g}^{q+1}}\left(x_k\right)$ and ${\hat \xi_{g \to k}^q}\left(x_k\right)$ have been approximated as Gaussian densities.
However, the computational complexity is still high when the system is large, since the number of messages scales with the number of potential UEs $K$.
In order to reduce the number of messages in the $q$-th iteration, we can further simplify the update rules by making some approximations.
Defining $Z_g^q = \sum\nolimits_{j=1}^K {s_{g,j}{\hat x_{j \to g}^q}}$, $V_g^q = \sum\nolimits_{j=1}^K {\left|s_{g,j}\right|^2v_{j \to g}^q}$, we can rewrite (\ref{Eq:Z_Mean}) as
\begin{equation}\label{Eq:Post Mean_Var}
\begin{array}{l}
\!\!\! Z_{g \to k}^q = Z_g^q - s_{g,k}{\hat x_{k \to g}^q},\ V_{g \to k}^q = V_g^q - \left|s_{g,k}\right|^2v_{k \to g}^q. \!\!\!
\end{array}
\end{equation}
Substituting (\ref{Eq:Post Mean_Var}) into (\ref{Eq:Var_Message_Approx}) and ignoring terms approximated as 0 in the large system limit $K \to \infty$ \cite{{AMP}}, the update rules at the variable nodes and the factor nodes can be approximated as in (\ref{Eq:Var_Update1})-(\ref{Eq:Fac_Update2}).
Finally, the posterior distribution of $x_k$ is given by (\ref{Eq:Post_Approx1}).

\vspace{-1.5mm}
\section{Proof of the Proposition 2}
\label{Proof:Belief Indicator}

If a reliable estimate of ${\bf X}_p$ is acquired after the convergence of the GMMV-AMP algorithm, the variance of the posterior distribution of $x_{k,m}$ tends to be zero, i.e., $v_{k,m}^{\infty} \to 0$, thus,
\begin{equation}
V_{g,m}^\infty = \mathop{\lim}\limits_{v_{k,m}^\infty \to 0} \sum\nolimits_k {\left|s_{g,k}\right|^2v_{k,m}^\infty}  = 0.
\end{equation}
Therefore, $D_{k,m}^\infty$ in (\ref{Eq:Var_Update1}) is calculated as
\begin{equation}\label{Eq:D_Infty}
D_{k,m}^\infty = \frac{\sigma}{\sum\nolimits_g {\left|s_{g,k}\right|}^2} \mathop \approx \limits^{(a)} \frac{\sigma}G.
\vspace{-0.5mm}
\end{equation}
Here, approximation $(a)$ is because the pilots are generated from an i.i.d standard complex Gaussian distribution, i.e., $s_{g,k} \sim {\cal CN}\left(s_{g,k}; 0, 1\right)$, thus $\sum\nolimits_g{\left|s_{g,k}\right|^2} \approx \sum\nolimits_g{\mathbb E}\left[\left|s_{g,k}\right|^2\right] = G$.
In the large system limit, as $G \to \infty$, $D_{k,m}^\infty \to 0^+$.
For a given realization of the massive access channel matrix ${\bf X}_p$, defining $r^\infty = \sum\nolimits_g {\frac{s_{g,k}^*\left(y_{g,m} - Z_{g,m}^\infty\right)}{\sigma} + V_{g,m}^\infty}$, $C_{k,m}^\infty$ in (\ref{Eq:Var_Update2}) is given as
\begin{equation}\label{Eq:C_Infty}
C_{k,m}^\infty  = \left\{ {\begin{array}{*{20}{c}}
  {D_{k,m}^\infty r^\infty {\text{,}}}&{{x_{k,m}} = 0,} \\
  {{x_{k,m}} + D_{k,m}^\infty r^\infty,}&{{x_{k,m}} \ne 0.}
\end{array}\begin{array}{*{20}{c}}
  {{\text{  }}} \\
  {{\text{   }}}
\end{array}} \right.
\end{equation}
Hence, substituting (\ref{Eq:D_Infty}) and (\ref{Eq:C_Infty}) into (\ref{Eq:L}), for $x_{k,m} = 0$, we have
\begin{equation}
{\cal L} = \mathop {\lim}\limits_{D_{k,m}^\infty  \to {0^+}} \left({\frac{1}{2}\ln \frac{{D_{k,m}^\infty }}{{D_{k,m}^\infty  + \tau }}} \right) - \frac{{{{\left| \mu  \right|}^2}}}{{2\tau }} = -\infty,
\end{equation}
while for $x_{k,m} \ne 0$,
\begin{equation}
\begin{aligned}
{\cal L} &= \mathop {\lim }\limits_{D_{k,m}^\infty  \to {0^ + }} \frac{1}{2}\left( {\ln \frac{{D_{k,m}^\infty }}{{D_{k,m}^\infty  + \tau }} + \frac{{{x_{k,m}}}}{{D_{k,m}^\infty }}} \right) - \frac{{{{\left| \mu  \right|}^2}}}{{2\tau }} \\
         &= \mathop {\lim }\limits_{D_{k,m}^\infty  \to {0^ + }} \frac{{{x_{k,m}}}}{{2D_{k,m}^\infty }} = +\infty,
\end{aligned}
\end{equation}
which yields
\begin{equation}
\pi_{k,m}^\infty = \left\{ {\begin{array}{*{20}{c}}
  {1,{\text{  }}{x_{k,m}} \ne 0}, \\
  {0,{\text{  }}{x_{k,m}} = 0}.
\end{array}} \right.
\end{equation}

\vspace{-1.5mm}
\section{Proof of the Proposition 3}
\label{Proof:SE}

Here, for simplicity of derivation, we focus on the $m$-th subgraph only and drop the antenna index $m$, as in Appendix \ref{Proof:Decouple Post}.
The derivation is based on (\ref{Eq:CS_Model_Spa}), and can be easily extended to model (\ref{eq:CS Model Vir}), thus we have
\begin{equation}\label{Eq:Salar_Measure}
\begin{array}{l}
y_g = \sum\nolimits_j s_{g,j}x_j + n_g.
\end{array}
\end{equation}
Substituting (\ref{Eq:Z_Mean}) and (\ref{Eq:Salar_Measure}) into (\ref{Eq:D and C}), $C_k^q$ is computed as
\begin{equation}\label{Eq:SE_C}
\begin{aligned}
C_k^q &= \frac{{\sum\nolimits_g{\frac{{{\left| {{s_{g,k}}} \right|}^2{x_k} + s_{g,k}^*{n_g} + s_{g,k}^*\sum\nolimits_{j \ne k}{s_{g,j}}\left({{x_j} - {{\hat x}_{j \to g}^q}}\right) }}{{\sigma  + \sum\nolimits_{j \ne k}{{\left| {{s_{g,j}}} \right|}^2v_{j \to g}^q}}}}}}{{\sum\nolimits_g {\frac{{{\left| {{s_{g,k}}} \right|}^2}}{{\sigma  + \sum\nolimits_{j \ne k} {{{\left| {{s_{g,j}}} \right|}^2}v_{j \to g}^q}}}}}}\\
      &\mathop \approx \limits^{(b)} \frac{{\sum\nolimits_g \!\! \left[{{{\left| {{s_{g,k}}} \right|}^2}{x_k} + s_{g,k}^*{n_g} + s_{g,k}^* \! \sum\nolimits_{j \ne k} {{s_{g,j}} \! \left( {{x_k} - {{\hat x}_{j \to g}^q}} \right)}}\right]  }}{{\sum\nolimits_g {{{\left| {{s_{g,k}}} \right|}^2}}}}\\
      &\mathop  \approx \limits {x_k} + \frac{{{\sum\nolimits_g {s_{g,k}^*{n_g} + \sum\nolimits_g {s_{g,k}^*\sum\nolimits_{j \ne k} {{s_{g,k}}\left( {{x_j} - {{\hat x}_{j \to g}^q}} \right)} } } } }}{G}.
\end{aligned}
\end{equation}
In (\ref{Eq:SE_C}), approximation $(b)$ is because the term $\sigma + \sum\nolimits_{j \ne k}{\left|s_{g,j}\right|^2}v_{j \to g}^q$ is approximately independent of $g$ in the large system limit with $K \to \infty$, which has been proven in \cite{{AMP}}.
Define
\begin{equation}\label{Eq:r_k}
\begin{array}{l}
r_k^q = \sum\nolimits_g{s_{g,k}^*n_g} + \sum\nolimits_g{s_{g,k}^*\sum\nolimits_{j \ne k}s_{g,k}\left(x_k-{\hat x_{j \to g}^q}\right)}.
\end{array}
\end{equation}
Since $s_{g,k} \sim {\cal CN}\left(s_{g,k}; 0, 1\right)$, $n_g \sim {\cal CN}\left(n_g; 0, \sigma\right)$, and $s_{g,k}$ is independent of $n_g$, $r_k$ follows a complex Gaussian distribution according to the central limit theorem as $K \to \infty$.
Moreover, by substituting (\ref{Eq:SE_MSE}) into (\ref{Eq:r_k}), we find the mean and the variance of $r_k^q$ are 0 and $G(\sigma+Ke^q)$, respectively.
Hence,
\begin{equation}\label{Eq:C_in_SE}
\begin{array}{l}
C_k^q = x_q + \sqrt{\frac{\sigma + Ke^q}G}z,
\end{array}
\end{equation}
where $z \sim {\cal CN}\left(z; 0, 1\right)$.
Meanwhile, by substituting (\ref{Eq:SE_Var}) into $D_k^q$ in (\ref{Eq:D and C}), we can obtain
\begin{equation}\label{eq:D_in_SE}
D_k^q {\approx} \left[\frac{\sum\nolimits_g{\left|s_{g,k}\right|^2}}{\sigma + \sum\nolimits_{j \ne k}{\left|s_{g,j}\right|^2}v_{j \to g}^q}\right]^{-1} \!\!= \frac{\sigma+K{\vartheta^q}}{G}.
\end{equation}
Hence, the Proposition 2 is proven.

\ifCLASSOPTIONcaptionsoff
\newpage
\fi

\end{document}